\documentclass[aps,pra,twocolumn,notitlepage,superscriptaddress]{revtex4-2} 

\usepackage{qcircuit}
 
\usepackage{algorithm}
\usepackage{algpseudocode}

\floatname{algorithm}{Protocol}
 
\usepackage{epsfig,amssymb,amsmath,mathrsfs,amsfonts,amsthm,graphicx,float,color}
\usepackage{caption}
\usepackage{subcaption} 

\usepackage[colorlinks=true]{hyperref}
\usepackage[active]{srcltx}
\usepackage{qcircuit} 
 \def\map#1{\mathcal #1}
\def\d{\operatorname{d}}\def\<{\langle}\def\>{\rangle}

\def\Tr{\operatorname{Tr}}\def\:{\hbox{\bf
    :}}

\def\R{\mathbb R}
\def\N{\mathbb N}

\def\grp#1{\mathsf{#1}}

\def\Span{\mathsf{Span}}
\def\Supp{\mathsf{Supp}}

\def\spc#1{\mathcal{#1}}

\def\poly{\mathsf{poly}}

\def\sc{{\rm sc}}

\def\set#1{\mathsf{#1}}

\newtheorem{theo}{{Theorem}}
\newtheorem{lemma}{{Lemma}}
\newtheorem{prop}{{Proposition}}

\newtheorem{remark}{{Remark}}
\newtheorem{defi}{{Definition}}

\begin{document}
\title{Compression of quantum shallow-circuit states}
 
\author{Yuxiang Yang}
\affiliation{ QICI Quantum Information and Computation Initiative, Department of Computer Science, The University of Hong Kong, Pokfulam Road, Hong Kong SAR, China}

\begin{abstract}
    Shallow quantum circuits feature not only computational advantages over their classical counterparts but also cutting-edge applications.
    Storing quantum information generated by shallow circuits is a fundamental question of both theoretical and practical importance that remained largely unexplored. 
    In this work, we show that $N$ copies of an unknown $n$-qubit state generated by a fixed-depth circuit can be compressed into a hybrid memory of $O(n \log_2 N)$ (qu)bits, which achieves the optimal scaling of memory cost.
    Our work shows that the computational complexity of resources can significantly impact the rate of quantum information processing, offering a unique and unified view of quantum Shannon theory and quantum computing.
\end{abstract}
 
\maketitle


\noindent{\it Introduction.} Shallow quantum circuits are a focus of recent research, for they are arguably the most accessible resources with genuine quantum features and advantages. Fundamentally, shallow quantum circuits with constant depth are hard to simulate classically (unless BQP $\subseteq$ AM) \cite{terhal2002adaptive}, and they outperform their classical counterparts in certain computational tasks \cite{bravyi2018quantum,bravyi2020quantum}.
Practically, variational shallow circuits \cite{benedetti2019parameterized,cerezo2021cost,skolik2021layerwise,abbas2021power} will remain a core ingredient of quantum algorithms in the noisy and intermediate-scale quantum (NISQ) era \cite{preskill2018quantum}. 
Efficient methods of learning shallow and bounded-complexity quantum circuits have recently been proposed \cite{zhao2023learning,yu2023learning,huang2024learning}.

Here we ask a fundamental question: Given $N$ copies of an {\em unknown} $n$-qubit state and the promise that it is generated by a shallow circuit, is there a faithful compression protocol that encodes the $N$-copy state into a memory of fewer (qu)bits and then decodes it up to an error vanishing at $N\to\infty$? 
Processing quantum states in the many-copy form is important for extracting, storing, and distributing quantum information.
Tasks where many-copy states serve as a fundamental resource, to list a few, include quantum metrology \cite{giovannetti2004quantum,giovannetti2006quantum,giovannetti2011advances}, state tomography \cite{d2001quantum,d2003quantum} and shadow tomography \cite{aaronson2018shadow,huang2020predicting}, quantum cloning \cite{gisin1997optimal,werner1998optimal,bruss1998optimal,scarani2005quantum}, and hypothesis testing \cite{helstrom1969quantum,yuen1975optimum,1976quantum,kholevo1979asymptotically}.
Quantum algorithms such as quantum principle component analysis \cite{lloyd2014quantum} also require states in the many-copy form.
As such, compression of many-copy quantum states is a basic and crucial protocol required for their storage and transmission.
In the literature, compression of many-copy states was first studied for pure qubits \cite{plesch2010efficient}, experimentally demonstrated in Ref.~\cite{rozema2014quantum}, and later generalized in a series of works to mixed qudits \cite{yang2016prlefficient,yang2016prlyang,yang2018prsaquantum,yang2018titcompression}.
However, regarding states generated by shallow quantum circuits, 
the existing results are not applicable, for they all assume the state to be in a fixed-dimension space. Here, instead, we consider states in a growing-dimension ($D=2^n$) space. 
Therefore, studying the compression of shallow-circuit states not only requires a better understanding of this important family of states but also demands new techniques of asymptotic quantum information processing.

\begin{figure}[bt] 
    \includegraphics[width=\linewidth]{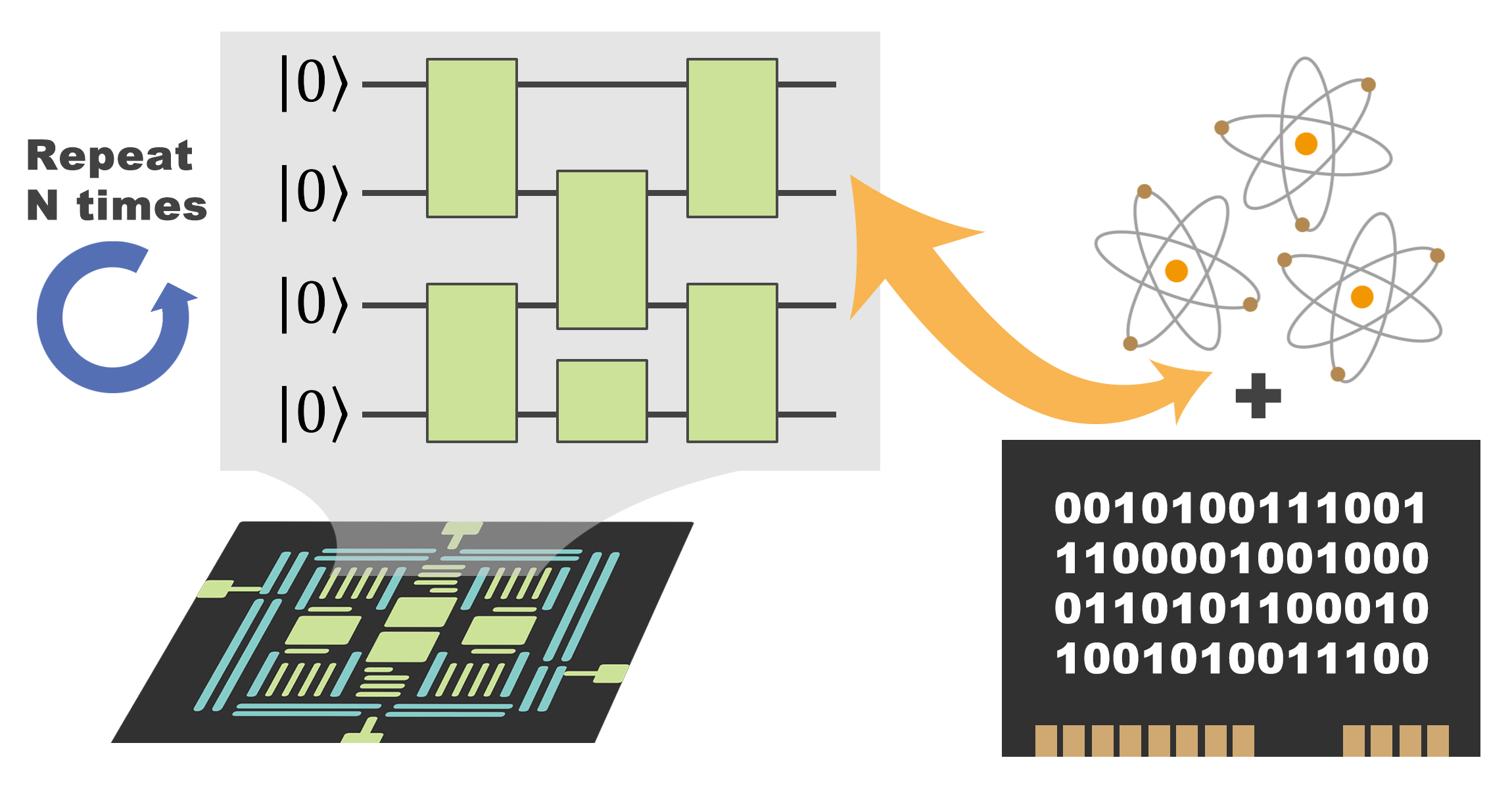}
    \caption{We show that $N$ copies of an $n$-qubit state generated by an unknown brickwork shallow circuit can be compressed into a classical-quantum hybrid memory of size $O(n\log_2 N)$ with vanishing error.} 
    \label{fig:compression}
\end{figure}

Without compression, storing the input state requires $N\cdot n$ qubits. Our main contribution, as illustrated in Figure \ref{fig:compression}, is to show that a faithful $N$-copy compression exists for brickwork shallow-circuit states, as long as $N$ grows at least as a polynomial of $n$ with a high enough degree. The memory cost of the compression is linear in $n$ and logarithmic in $N$, achieving the optimal rate with an exponential memory reduction in terms of $N$. Intriguingly, the memory does not have to be fully quantum. Instead, a classical-quantum hybrid memory suffices, where the ratio between the number of qubits and the number of classical bits decreases as $O(\log_2 n/\log_2 N)$. That is to say, when $N$ is large, the memory consists mainly of classical bits, while we also prove that a fully classical memory does not work.
The compression is achieved with novel tools, which we develop for quantum information processing in the asymptotic regime of many copies. They can be further applied to other information-processing tasks involving complexity-constrained quantum states.


\smallskip

\noindent{\it The setup.} 
For a pure state $|\psi\>$, $\psi$ denotes the projector $|\psi\>\<\psi|$.  For a unitary $U$, the corresponding calligraphic letter $\map{U}$ denotes the corresponding channel $\map{U}(\cdot):=U(\cdot)U^\dag$.
The big-$O$, big-$\Omega$, and big-$\Theta$ notation \footnote{For two real-valued functions $f,g$, $f=O(g)$ if there exist $n_0$ and $c>0$ such that $|f(n)|\le c\cdot g(n)$ for every $n\ge n_0$, $f=\Omega(g)$ if $g=O(f)$, and $f=\Theta(g)$ if both $f=O(g)$ and $f=\Omega(g)$. In addition, we denote by $f\ll g$ (which is also denoted by $f=o(g)$ in the literature) for non-negative functions $f$ and $g$ if $f(n)/g(n)$ vanishes in the limit of large $n$.} will be used to discuss the asymptotic behavior of parameters.
A function $f=f(n)$ is polynomial in $n$, denoted by $f\in\poly(n)$, if the exists a constant $k\ge 0$ such that $f=O(n^k)$.
 
Given a set $\set{S}$ of quantum states, the task of faithful $N$-copy compression is to design a protocol that consists of an encoder $\map{E}_N$ and a decoder $\map{D}_N$ such that the compression error vanishes for large $N$:
\begin{align}\label{faithfulnesscompress}
\lim_{N\to\infty}\sup_{\rho\in\set{S}}d_{\Tr}\left(\map{D}_N\circ\map{E}_N(\rho^{\otimes N}),\rho^{\otimes N}\right)=0.
\end{align}
Here $d_{\Tr}$ denotes the trace distance between quantum states. We focus on the blind setting \cite[Chapter 10]{hayashi2016quantum}, where $\map{E}_N$ and $\map{D}_N$ are dependent on $N$ but are independent of the input state. The goal of compression is to reduce the memory cost 
$    M:=\log_2\left|\Supp\left\{\map{E}_N(\rho^{\otimes N})\right\}_{\rho\in\set{S}}\right|$,
i.e., the number of (qu)bits required for storing $\map{E}_N(\rho^{\otimes N})$, while respecting the faithfulness condition (\ref{faithfulnesscompress}).

Here, we are interested in $n$-qubit pure states generated by circuits of depth no more than $d$. As a proof-of-principle example, we focus on the most representative case of (1D) brickwork shallow circuits and consider the set of \emph{shallow-circuit states} $\set{S}_{\sc}:=\left\{|\psi\>~:~|\psi\>=U_{\rm sc}|0\>^{\otimes n}\,\exists\,U_{\sc}\right\}$,
where $U_\sc$ is a brickwork circuit with bounded depth $(\le d)$ (as illustrated in Figure \ref{fig:compression}). We treat $n$ as the asymptotic parameters, $N$ as a growing function of $n$, and $d$ as a fixed parameter.

We begin with a sketch of the idea. 
First, by performing tomography on a negligible portion of the copies (while keeping most of the $N$ copies untouched), we \emph{localize} the compression problem by fixing the unknown shallow-circuit state to a neighborhood (i.e., a region with radius vanishing in $N$). We then show that all shallow-circuit states in this neighborhood can be generated by a circuit consisting of $O(n)$ small rotations plus a fixed shallow circuit, thereby locally parameterizing this neighborhood by $O(n)$ parameters. We proceed to establish an approximate equivalence between copies of a state within this neighborhood and a (multi-mode) coherent state of light, which can be compressed by photon number truncation. More details are as follows.
 
\smallskip

\noindent{\it Efficient local parameterization of shallow-circuit states.}
The first ingredient is a tool to pin the shallow-circuit state $|\psi_{\sc}\>$ down to a small local region parameterised by up to $\poly(n)$ real parameters. We call this task the \emph{local parameterization of quantum states}, which is achieved by obtaining a crude estimate and then constructing a local model that incorporates the remaining uncertainty of the state. The first step requires shallow-circuit state tomography, e.g., \cite{zhao2023learning}. The other step, however, is rather tricky as one has to parameterise the estimated neighborhood with as few parameters as possible.

We begin with general definitions. A unitary gate $U^\epsilon$ is called an $\epsilon$-rotation if $d_{\diamond}(\map{U}^\epsilon,\map{I})\le\epsilon$, where $d_{\diamond}$ denotes the diamond distance. 
An $\epsilon$-rotation can be expressed as a short-time unitary evolution $e^{-i2\epsilon H}$ for some Hermitian $H$ with $\|H\|_{\infty}\le 1$ (see Appendix for the proof). 
\begin{defi}
    Given a set $\set{S}$ of $n$-qubit quantum states and a fixed quantum state $|\psi_0\>=U_0|0\>^{\otimes n}$, with the guarantee that $d_{\Tr}(\psi,\psi_0)\le\epsilon$ for every $|\psi\>\in\set{S}$, the task of local parameterization of $\set{S}$ is to identify a set $\set{W}_{\rm loc}$ of $O(\epsilon)$-rotations, such that for every $|\psi\>\in\set{S}$ there exists a $W\in\set{W}_{\rm loc}$ with $|\psi\>=U_0W|0\>^{\otimes n}$. Moreover, the local parameterization is efficient if $\set{W}_{\rm loc}$ can be parameterized with $\poly(n)$ local parameters, i.e., if there exists an on-to function from $\set{E}^{f}$ to $\set{W}_{\rm loc}$ with $\set{E}\subset\R$ being an $O(\epsilon)-$ length interval and $f\in\poly(n)$.
\end{defi}
It is quite easy to find a local parameterization: Applying $(U_0)^\dag$ on every $|\psi\>$, the similarity of $U_0^\dag|\psi\>$ and $|0\>^{\otimes n}$ guarantees the existence of an $O(\epsilon)$-rotation. However, achieving an efficient parameterization is nontrivial. For an $n$-qubit state, the above naive strategy leaves the state in a neighborhood with exponentially many parameters.
A hope is to improve parameter efficiency when the set of states under consideration has a certain structure. For instance, when $\set{S}$ consists of states generated by locally acting an unknown quantum gate on a few qubits of a fixed $n$-qubit state $|\psi_0\>$.
Suppose we know that all states in $\set{S}$ are close to $|\psi_0\>$. It is quite natural to speculate that we could replace the few-qubit gate by an $O(\epsilon)$-rotation, acting on the same qubits, to achieve an efficient local parameterization (since the $O(\epsilon)$-rotation acts only on a few qubits, it can be modeled by fewer parameters).
Somewhat surprisingly, this idea fails: There exists a bipartite pure state $|\psi\>_{AB}$ and a unitary $V_{B}$ acting on $B$ such that $1)$ $d_{\Tr}((\map{I}_A\otimes\map{V}_{B})(\psi_{AB}),\psi_{AB})\le\epsilon$ and $2)$ there is no $O(\epsilon)$-rotation $V_{B}^\epsilon$ such that $(\map{I}_A\otimes\map{V}_{B})(\psi_{AB})=(\map{I}_A\otimes\map{V}_{B}^\epsilon)(\psi_{AB})$ (see 
\footnote{One such example is  \unexpanded{$|\psi\rangle_{AB}=\sqrt{1-\epsilon}|0\rangle_{A}|00\rangle_{B}+\sqrt{\epsilon}|1\rangle_A\frac{|10\rangle_B+|20\rangle_B}{\sqrt{2}}$} (where $A$ is a qubit system and $B$ consists of two qutrit systems) and a gate $V_{B}$ that acts trivially on the second qutrit if the first qutrit is in \unexpanded{$|0\rangle$} and maps the second qutrit from  \unexpanded{$|0\rangle$} to \unexpanded{$|1\rangle$} (\unexpanded{$|2\rangle$)} conditioning on the first qutrit being  \unexpanded{$|1\rangle$} (\unexpanded{$|2\rangle$)}. The output state is   \unexpanded{$\sqrt{1-\epsilon}|0\rangle_A|00\rangle_B+\sqrt{\epsilon}|1\rangle_A|B_0\rangle_{B}$}
(with \unexpanded{$|B_0\rangle:=\frac{|11\rangle+|22\rangle}{\sqrt{2}}$)}. If a desired $O(\epsilon)$-rotation $V^\epsilon$ on $B$ exists, it has to take the initial marginal state on $B$ to the marginal of the output state, i.e.,
 \unexpanded{$(1-\epsilon)|00\rangle\<00|+\epsilon\left(\frac{|10\rangle+|20\rangle}{\sqrt{2}}\right)\left(\frac{\<10|+\<20|}{\sqrt{2}}\right)\xrightarrow{V^\epsilon}(1-\epsilon)|00\rangle\<00|+\epsilon|B_0\rangle\<B_0|$}.
Since a unitary cannot change the eigenvalues, we need
 \unexpanded{$V^\epsilon\left(\frac{|1\rangle+|2\rangle}{\sqrt{2}}\otimes|0\rangle\right)=|B_0\rangle$}.
Note that the initial state is a product state while the final state has entanglement bounded away from zero, this transformation cannot be achieved by a unitary gate close to the identity.} for a concrete example).

The above remark shows that an efficient local parameterization is not guaranteed even if the set of states is parameterized by only a constant number of parameters.
As such, our goal -- an efficient local parameterization of brickwork shallow-circuit states -- is not as simple as it appears. We first use a lightcone argument to reduce the generating circuit into two effective layers, which is estimated up to a small error $\epsilon$. 
We then apply the inverse of the estimated outer layer gates to the state. Each gate in the outer layer of the generating circuit gets canceled out, up to an $O(\epsilon)$-rotation acting on a constant number of adjacent qubits. These $O(\epsilon)$-rotations are overlapping, but only with a constant of neighbors due to the locality of the shallow circuit. Therefore, the impact of each $O(\epsilon)$-rotation is limited to a constant-size neighborhood, and we can fully capture the impact by another $O(\epsilon)$-rotation with a constant number of parameters. Since there are only $O(n)$ such $O(\epsilon)$-rotations, putting these all together we get a parameter-efficient localization of the outer layer. A similar treatment can be applied to the inner layer. All together we get the desired result
(see Appendix for details):
\begin{theo}
\label{theo1}
Let $\set{S}_{\sc}$ be the subset of $n$-qubit states generated by depth-$d$ (1D) brickwork circuits. Given a fixed shallow-circuit state $|\psi_0\>=U_0|0\>^{\otimes n}$ of the same structure and the guarantee $d_{\Tr}(\psi_{\sc},\psi_0)\le\epsilon$ for every $|\psi_{\sc}\>\in\set{S}_{\sc}$, there exists a set $\set{W}_{\rm loc}$ of $O(\epsilon)$-rotations, parameterized by $O(n\cdot (2^{16d}/d))$ (real) local parameters, such that for every $|\psi_{\sc}\>\in\set{S}_\sc$ there exists a $W\in\set{W}_{\rm loc}$ with $|\psi_\sc\>=U_0W|0\>^{\otimes n}$. Here $W=\prod_{j=1}^{n_{\rm gate}}W^{28\epsilon}_j$
with $n_{\rm gate}\le n/d+4$, and each $W^{28\epsilon}_j$ is a $(28\epsilon)$-rotation acting on no more than $8d$ neighboring qubits. 
On each qubit, no more than $n_{\rm overlap}=4$ gates in  $\{W_j^{28\epsilon}\}$  act non-trivially.  
\end{theo}  

\smallskip

\noindent{\it Quantum local asymptotic normality (Q-LAN).}
Generally speaking, taking $N$ copies of a $D(<\infty)$-dimensional state as input, a Q-LAN \cite{guctua2006local,guctua2007local,kahn2009local} is an input-independent transformation that outputs a state close to a multi-mode Gaussian state. It is approximately reversible, as there exists an inverse Q-LAN recovering the original state up to an error vanishing in $N$. Here locality means that the Q-LAN works only if the concerned states are in a neighborhood with a vanishing radius.
The processing of many-copy finite-dimensional quantum systems can then be reduced to that of Gaussian systems, which has been a key ingredient to achieve the ultimate precision limit of multiparameter quantum state estimation \cite{gill2013asymptotic,yang2019cmpattaining,girotti2023optimal}.

For our task, we would like to reduce $N$ copies of a locally-parameterized shallow-circuit state to a Gaussian state via the Q-LAN.
A crucial issue with the existing Q-LANs is they require $D$ to be fixed with respect to $N$ ($n$), whereas we need $D$ to be growing exponentially fast in $n$. In addition, for existing Q-LANs, the parameterization of the $D$-dimensional state is set to a specific form  (see \cite{kahn2009local}) that does not match ours.
In the following, we bridge the gaps by proving a \emph{new} Q-LAN.

We prove the result in a more general form that covers the reduced state in Theorem \ref{theo1}.
First, define a circuit template to be $Q=(q_1,\dots,q_G)$, with each $q_j\subset[n]$ and $|q_j|\le \tilde{d}$ for some $\tilde{d}$ independent of $n$. Let 
$\set{S}_{N,\eta}(Q):=\left\{|\psi(\vec{H})\>^{\otimes N}~:~|\psi(\vec{H})\>=\prod_{j=1}^{G}e^{-i\frac{\eta}{\sqrt{N}} H_j}|0\>^{\otimes n}\right\}$
be the $N$-copy family of $n$-qubit pure states generated by $O(\eta/\sqrt{N})$-rotations with a fixed $Q$, where every $H_j$ is a Hermitian operator acting trivially on qubits not in $q_j$ and $\|H_j\|_{\infty}\le 1$.  
Second, define $\set{B}(Q)=\cup_{j=1}^G\left\{\vec{k}\in2^{[n]}~:~k_i=0~\forall i\not\in q_j\right\}$ that keeps track of all $n$-qubit strings that could appear when acting $\sum_j H_j$ upon $|0\>^{\otimes n}$, and by
\begin{align}\label{cohstate}
    |\vec{u}(\vec{H})\>_{\rm coh}=\bigotimes_{\vec{k}\in\set{B}(Q)}|u^{\vec{k}}_1+iu^{\vec{k}}_2\>_{\rm coh}
\end{align}
a coherent state with no more than $G\cdot 2^{\tilde{d}}$ modes, where $|u\>_{\rm coh}:=e^{-|u|^2/2}\sum_{m=0}^{\infty}\frac{u^m}{\sqrt{m!}}|m\>$ is a (single-mode) coherent state in an infinite-dimensional space spanned by the Fock basis $\{|m\>\}$, and $u^{\vec{k}}_1$ ($u^{\vec{k}}_2$) is the real (imaginary) part of $\eta\<\vec{k}|\sum_{j=1}^G H_j|\vec{0}\>$. The amplitude of each mode of $|\vec{u}(\vec{H})\>_{\rm coh}$ is upper bounded by 
\begin{align}\label{amplitudebound1}
    \max_{\vec{k}}|u^{\vec{k}}_1+iu^{\vec{k}}_2|\le \eta\cdot n_{\rm overlap}
\end{align}
where $n_{\rm overlap}:=\max_{i\in[n]}\sum_{j:i\in q_j}1$ is the maximum of the number of subsets overlapping at any local site.
Our new Q-LAN is stated as follows:
\begin{theo}
\label{theo2}
When $\eta^2 G^2\ll \sqrt{N}$, there exist quantum channels $\map{V}_N$ (the Q-LAN) and $\map{V}^\ast_N$ (the inverse Q-LAN) such that  any $|\psi(\vec{H})\>\in\set{S}_{N,\eta}(Q)$ can be reversibly converted into a coherent state (\ref{cohstate}) with no more than $G\cdot 2^{\tilde{d}}$ modes up to errors vanishing in $N$.
\end{theo}
To establish the new Q-LAN, we first observe that the state can be expanded as $|\psi(\vec{H})\>=|0^n\>+\sum_{\vec{p}\in\set{B}(Q)\setminus\{0^n\}}(\eta/\sqrt{N})a_{\vec{p}}|b\>+O(\eta^2/N)$ for some bounded $\{a_{\vec{p}}\}$.
If $\eta^2/N$ is as small as assumed in Theorem \ref{theo2}, the difference between $|\psi(\vec{H})\>$ and $|\tilde{\psi}(\vec{H})\>\propto |0^n\>+\sum_{\vec{p}\in\set{B}(Q)\setminus\{0^n\}}(\eta/\sqrt{N})a_{\vec{p}}|b\>$ will be negligible, so we can focus on constructing a Q-LAN for $|\tilde{\psi}(\vec{H})\>^{\otimes N}$. 
Note that $|\tilde{\psi}(\vec{H})\>$ lives in a much smaller space of dimension $|\set{B}(Q)|=O(G)$.
It follows that we may express $|\tilde{\psi}(\vec{H})\>^{\otimes N}$ in the symmetric basis, and define $\map{V}_N$ to be one that maps symmetric states to multi-mode Fock states according to their types.
More details can be found in Appendix.

\smallskip

\begin{algorithm}[H]
  \caption{Compression protocol for $N$-copy shallow-circuit states.} \label{protocol:compression}
   \begin{algorithmic}[1]
   \Statex {\em Encoder:}
   \Statex {\bf Input:}  $N$ copies of a $n$-qubit state $|\psi_\sc\>$ generated by a brickwork circuit of depth-$d$; 
    \Statex a configuration of parameters $(\epsilon_0, N_0,\alpha_0)$.
   \Statex {\bf Require:} A classical memory of $M_{\rm c}$ bits and a quantum memory of $M_{\rm q}$ qubits.
   \State Construct an $\epsilon_0$-covering set $\set{U}_{\sc}(\epsilon_0)$ for brickwork circuits of depth-$d$. That is, for every $|\psi_{\sc}\>$, there exists an element of $\set{U}_{\sc}(\epsilon_0)$ whose trace distance with $|\psi_\sc\>$ is upper bounded by $\epsilon_0$.
   \State Run state tomography with $N_0$ copies of $|\psi_\sc\>$, which outputs an estimate $|\hat{\psi}_{\sc}\>=\hat{U}|0\>^{\otimes n}$ with $\hat{U}\in\set{U}_{\sc}(\epsilon_0)$. Store $\hat{U}$ in the classical memory.
   \item Apply first the inverse of $\hat{U}^{\otimes (N-N_0)}$ and then the Q-LAN transformation on the remaining $N-N_0$ copies.
   \item   Amplify each mode of the resultant multi-mode coherent state: $|z\>\to |\sqrt{N/(N-N_0)}z\>$, i.e., with intensity gain $N/(N-N_0)$.
   \item   Compress the multi-mode coherent state via truncating each mode to less than $(e\alpha_0)^2$ photons. Store the resultant state in the quantum memory. 
   \end{algorithmic} 
    \begin{algorithmic}[1]
   \Statex {\em Decoder:} 
   \Statex {\bf Require from the encoder:} both the quantum memory (step 5) and the classical memory (step 2).
   \item Apply the inverse Q-LAN transformation on the state of the quantum memory.
   \item Retrieve $\hat{U}$ from the classical memory and apply $\hat{U}^{\otimes N}$. 
   \item {\bf Output} the final state.
   \end{algorithmic}
\end{algorithm}

\noindent{\it The compression protocol.}
With both tools in place, we are now ready to define our compression protocol, shown as Protocol \ref{protocol:compression}.

In the End Matter, we give a high-level analysis of the protocol's error and its memory cost, leaving details to Appendix. Essentially, we show that Protocol \ref{protocol:compression} with a configuration $
N_0=\Theta(N^{1-\frac{1}{2}\Delta}), \epsilon_0 =\Theta(n\cdot N^{-\frac12(1-\frac23\Delta)})$, and $\alpha_0 =\Theta(\sqrt{N}\epsilon_0)$
achieves both faithfulness and the desired compression rate, where $\Delta>0$ is a coefficient independent of $n$ and $N$.Another important question besides the memory efficiency is the computational efficiency. 
In Appendix, we show that Protocol \ref{protocol:compression} can be implemented within $\poly(n)$ time, via a high-dimensional Schur transform \cite{krovi2019efficient} and parallelizing the tomography \cite{zhao2023learning} using the brickwork structure of generating circuits. 
It is a meaningful future direction of research to further reduce the complexity and to make the compression executable on practical devices.

Detailed error analysis (in the End Matter) also yields that $N$ must grow at least as a polynomial of $n$ with a high enough degree.
Summarizing, we have the following theorem as the main result for compression:
\begin{theo}
\label{theo:compression}
Let $\set{U}_\sc$ be the collection of all $n$-qubit pure states generated by depth-$d$ (1D) brickwork circuits, with $d$ being a fixed parameter. Assume $N=\Theta(n^{\frac{32}{3}+\gamma})$ for some $\gamma>0$. For any  $\Delta\in\left(\frac{6}{32+3\gamma},\frac{3}{4}-\frac{18}{32+3\gamma}\right)$, there exist a time-efficient compression protocol that faithfully compresses $N$ copies of any state from $\set{U}_\sc$ into a hybrid memory of $8(1-2\Delta/3)nd\left(\log_2 N+O\left(\log_2 n\right)\right)$ classical bits and $(n/d+4)2^{8d+1}/3\cdot\left(\Delta\log_2 N+O(\log_2 n)\right)$ qubits.
\end{theo}

The total memory cost of our protocol is in the order of $n\cdot\log_2 N$.
One may ask whether the compression rate is optimal and if the cost can be made even smaller in terms of $n$ or $N$.
Via an information-theoretic approach (details can be found in Appendix), we show that any faithful protocol requires a memory size close to the Holevo information \cite{holevo1973bounds} of the ensemble of shallow-circuit states, which is $\Omega(n\cdot\log_2 N)$. This proves the optimality of the compression rate in the scaling of both $n$ and $N$.

The quantum-to-classical cost ratio (in the leading order) is 
$
    r_{\rm q-c}=\frac{2^{8d-2}\Delta}{d^2(3-2\Delta)}$
When the number of copies gets larger fixing $n$, i.e., when $\gamma\gg 1$, we may choose $\Delta\to 6/(32+3\gamma)$ and the ratio $r_{\rm q-c}=O(1/\gamma)$. This means that,
although the sheer number of qubits in the hybrid memory might be large, it becomes negligible when compared to the number of bits. In this sense, when $N$ is large enough, $N$-copy shallow-circuit states can be converted into classical bits plus a small portion of qubits. The compressed $N$-copy state can be regarded as a long (classical) file with a quantum signature appended to it, similar as in the case of general qudits \cite{yang2018titcompression}.

A reasonable suspicion then is whether any qubit is needed at all in the memory. A positive answer would render the compression much less interesting. Indeed, it would then be reduced to learning, where the encoder consists of tomography with all the $N$ copies and storing the outcome properly.
Here we give a negative answer to this question and show the inability of any incoherent protocol. Explicitly, we show that any protocol using a fully classical memory cannot be faithful, no matter how large the classical memory is. 
This justifies the unique position of the task of compression as a fully coherent task, distinguishing it from the incoherent task of learning \cite{zhao2023learning,yu2023learning,huang2024learning}.
The proof idea Appendix, originally appeared in \cite{yang2018titcompression} to prove a similar result for qudits, is to keep track of genuinely quantum monotones such as the trace distance and the quantum Hellinger distance \cite{shunlong2004informational}. On one hand, we can find shallow-circuit states for which these monotones are strictly different, and the difference must be preserved by any faithful compression. On the other hand, these monotones coincide when the states are collapsed into bits, and the difference cannot be recovered, leading to a contradiction.


\smallskip

\noindent{\it Applications and extensions.} 
Our compression can store multiple instances of the same program at intermediate steps of a constant-depth computation or transfer them more efficiently in distributed settings. 
The need to compress $N$-copy states is also naturally encountered in new scenarios of quantum sensing, when the goal is to estimate the \emph{total} effect of signals that occur in disjoint time intervals, e.g., the total duration of multiple spikes or the mean effect of intermittent fields. 
Our new protocol for shallow quantum circuits matches the trend to consider many-body probe states prepared by low-depth circuits \cite{marciniak2022optimal}. 
The task is also a natural analogy of distributed quantum metrology \cite{komar2014quantum,proctor2018multiparameter,ge2018distributed,liu2021distributed,kwon2022quantum,yang2024quantum}, where signals are now distributed in time instead of in space.
There using an array of $N$ identical probes and the $N$-copy state compression to store the acquired information coherently and memory-efficiently in between time intervals achieves a substantial accuracy enhancement
compared to measuring the signals separately, which has been verified in \cite{yang2018prsaquantum} for qubit probes used as a ``quantum stopwatch". 
Our protocol here enables multiparameter sensing using correlation-bounded many-body physical systems, which covers the detection of various electromagnetic properties.
The application can even go beyond metrology:
For example, the optimal tradeoff between the program size and the accuracy for  programming a constant-size quantum gate is achieved via metrology \cite{yang2020prloptimal}, which might be extended to the programming of (1D brickwork) shallow quantum circuits using our compression as a subroutine.

As we focused only on the most fundamental case, there is plenty of room for extension, e.g., to shallow-circuit states with a 2D structure. 
It is likely that we need to extend the tools, especially efficient local parameterization (Theorem \ref{theo1}), to 2D circuits. 
One possible route is to design new circuit covering schemes \cite{landau2024learning} for 2D or higher-dimension circuits that meet the requirement of local parameterization.
Moreover, the circuit depth $d$ is treated as a constant throughout this work, but from the derivation of results it can be seen that the compression will still be faithful when $d$ grows very slowly (e.g., $d\ll\log n$) with $n$.
In particular, it would be interesting to discuss pesudorandom quantum states \cite{ji2018pseudorandom}, which are low-depth states processing approximate Haar-randomness and are thus of particular interest in quantum cryptography.
At last, the extension to noisy setting is of particular interest in practice. While similar results are expected there, some techniques in this work do not immediately generalize to mixed states and require moderate adaptation.

This work serves as the first step in establishing a new direction of coherent quantum information processing where the complexity of resources determines the rate and performance of processing, which goes beyond the existing literature that focused on incoherent information processing  \cite{zhao2023learning,yu2023learning,huang2024learning}. 
For future perspectives, it is our goal to consider more tasks such as cloning \cite{gisin1997optimal,werner1998optimal,bruss1998optimal,scarani2005quantum} and gate programming \cite{nielsen1997programmable,ishizaka2008asymptotic,kubicki2019resource,sedlak2019optimal,yang2020prloptimal} and, ultimately, to re-examine the entire quantum Shannon theory established in the past decade from the new perspective of the NISQ era.

\smallskip

\noindent{\bf Acknowledgement.}
The author acknowledges Penghui Yao and Fang Song for helpful discussions on pseudorandom states and Xinhui Yang for making Figure 1. 
This work is supported by the National Natural Science Foundation of China via the Excellent Young Scientists Fund (Hong Kong and Macau) Project 12322516, Guangdong Provincial Quantum Science Strategic Initiative (GDZX2303007 and GDZX2203001), Guangdong Basic and Applied Basic Research Foundation (Project No.~2022A1515010340),
and the Hong Kong Research Grant Council (RGC) through the Early Career Scheme (ECS) grant 27310822 and the General Research Fund (GRF) grant 17303923.

\medskip

\noindent{\bf End Matter on the error and the memory cost of Protocol \ref{protocol:compression}.}
Here we provide an analysis of the total error and the memory cost of our compression protocol.

The total error will be vanishing if the error of each subroutine vanishes.
First, the tomography, with an error that scales as
$1/\sqrt{N_0}\ll\epsilon_0$,
determines an estimate $\hat{U}|0\>^{\otimes n}$ such that  $|\psi_{\sc}\>$ is within its $\epsilon_0$-neighborhood with high chance. Applying the inverse of $\hat{U}$ thereby localizes the remaining copies to a neighborhood with only $\poly(n)$ parameters. By Theorem \ref{theo1}, the generating circuit of the localized state consists of $O(n)$ $\epsilon_0$-rotations. Via the Q-LAN and the inverse Q-LAN (cf.~Theorem \ref{theo2}), we can interchange between the remaining $N-N_0$ copies of the shallow-circuit state and a multi-mode coherent state $|\vec{u}\>_{\rm coh}$.
The coherent state $|\vec{u}\>_{\rm coh}$ has $O(n)$ modes, whose amplitudes' moduli are all upper bounded by $\eta=O(\sqrt{N}\epsilon_0)$.
By Theorem \ref{theo2}, the Q-LAN and the inverse Q-LAN have vanishing errors, as long as  
$\sqrt{N}\gg n^2 N\epsilon_0^2$. 
The copies consumed by tomography are translated into the decrease of the coherent state's amplitude by a ratio of $\sqrt{\frac{N}{N-N_0}}$, which can be compensated by amplitude amplification \cite{caves1982quantum}. 
The idea here is similar as that of approximate and asymptotic cloning of quantum states. Although quantum states cannot be cloned perfectly \cite{wootters1982single,dieks1982communication}, they can be cloned approximately. In particular, for the task of generating $N_0$ extra copies out of $N-N_0$ input copies, the infidelity of cloning vanishes if $N\gg N_0\gg 1$ \cite{werner1998optimal}, which is the case here.
Now, the coherent state has $O(n)$ modes. For each mode, we apply an individual amplifier. The total amplification error is ensured to be vanishing as long as the single amplifier's error, which scales as $N_0/N$, is small compared to $1/n$.
At last, each mode of the coherent state can be compressed by a photon number truncation. As the photon number of a single-mode coherent state $|u\>_{\rm coh}$ follows the Poisson distribution with mean $|u|^2$, it is enough to truncate at a suitable constant times the mean to ensure faithfulness, which can be satisfied by choosing $\alpha_0^2$ to be a suitable constant times $\epsilon_0^2 N$.

As for the memory cost, it is stated in the protocol that the hybrid memory consists of two parts: a classical memory to store $\hat{U}$ and a quantum memory to store the truncated coherent state. The classical memory cost equals $\log_2|\set{U}_\sc(\epsilon_0)|$ bits, where $|\set{U}_\sc(\epsilon_0)|$ denotes the cardinality of the $\epsilon_0$-covering set.
Intuitively, for a circuit of $G$ two-qubit gates, we can construct an $\epsilon_0$-covering by concatenating $G$ $(\epsilon_0/G)$-coverings of two-qubit gates. Each $\epsilon$-covering of two-qubit gates can be stored with $O(\log_2(1/\epsilon))$ bits, and thus the cost for the $G$ gate covering scales as $G\log_2(G/\epsilon_0)$. In our case, $G=O(n)$ for shallow circuits and, since $\epsilon_0=\Theta(n\cdot N^{-\frac12(1-\frac23\Delta)})$, the size of the classical memory is  $O(n(1-\frac23\Delta)\log_2 N)$.
For the quantum memory, 
each mode of the coherent state has an amplitude upper bounded by $O(\sqrt{N}\epsilon_0)$, according to Eq.~(\ref{amplitudebound1}) and Theorem \ref{theo1}.
Since there are $O(n)$ modes and each mode's photon number is truncated at $O(\alpha_0^2)=O(n^2 N^\Delta)$, it requires $O(\Delta\log_2 N)$ qubits (assuming $N\gg n$) to store each mode, and $O(n\Delta\log_2 N)$ qubits in total.

We note that the aforementioned constraints ($\sqrt{N}\gg n^2 N\epsilon_0^2$ and $N_0/N\ll 1/n$) together set a constraint on the relation between $n$ and $N$, as state in Theorem \ref{theo:compression}.

\bibliographystyle{apsrev4-2}
\bibliography{ref}

\begin{widetext}

\appendix
 
\tableofcontents

\section{Preliminaries}\label{sec:prelim} 
\subsection{Conventions and notations}
We denote by $\spc{H}_k$ the $k$-dimensional Hilbert space, and $\spc{H}_2$ denotes the Hilbert space of a qubit.
For a pure state $|\psi\>$, we denote by $\psi$ the projector $|\psi\>\<\psi|$.  For a unitary $U$, we use the corresponding calligraphic letter $\map{U}$ to denote the corresponding channel $\map{U}(\cdot):=U(\cdot)U^\dag$.

We will make frequent use of the big-$O$, big-$\Omega$, and big-$\Theta$ notation. For two real-valued functions $f,g$, $f=O(g)$ if there exist $n_0$ and $c>0$ such that $|f(n)|\le c\dot g(n)$ for every $n\ge n_0$, $f=\Omega(g)$ if $g=O(f)$, and $f=\Theta(g)$ if both $f=O(g)$ and $f=\Omega(g)$. In addition, we denote by $f\ll g$ (which is also denoted by $f=o(g)$ in the literature) for non-negative functions $f$ and $g$ if $f(n)/g(n)$ vanishes in the limit of large $n$.
We say a function $f=f(n)$ is polynomial in $n$, which is denoted by $f\in\poly(n)$, if the exists a constant $k\ge 0$ such that $f=O(n^k)$. 

\subsection{Distance measures for quantum states and gates}
For two quantum states $\rho$ and $\sigma$ of the same quantum system, the trace distance between $\rho$ and $\sigma$ is defined as
\begin{align}\label{deftracedist}
    d_{\Tr}(\rho,\sigma):=\frac12\|\rho-\sigma\|_1,
\end{align}
where $\|A\|_1:=\Tr|A|$ is the trace norm of operators. The trace distance satisfies the triangle inequality, i.e., for any states $\rho,\sigma,$ and $\eta$, we have
\begin{align}
    d_{\Tr}(\rho,\eta)\le d_{\Tr}(\rho,\sigma)+d_{\Tr}(\sigma,\eta).
\end{align}
The most commonly used measure of similarity for quantum states is the fidelity:
\begin{align}
    F(\rho,\sigma)=\left(\Tr\sqrt{\rho^{\frac12}\sigma\rho^{\frac12}}\right)^2.
\end{align}
For pure states, the fidelity takes a simpler form 
$F(\psi,\phi)=|\<\psi|\phi\>|^2$.
A more intuitive expression of the fidelity is given by Uhlmann's theorem:
\begin{align}
    F(\rho,\sigma)=\max_{U}\left|\<\psi_\sigma|(I\otimes U_A)|\psi_\rho\>\right|^2,
\end{align}
where $|\psi_\rho\>,|\psi_\sigma\>\in\spc{H}\otimes\spc{H}_A$ are any purifications of $\rho$ and $\sigma$, and the maximization is over any unitary $U_A$ on an ancillary system $\spc{H}_A\simeq\spc{H}$.
The approximate equivalence of the trace distance and the fidelity can be established via the Fuchs-van de Graaff inequalities:
\begin{align}\label{fuchsvandegraaff}
    1-\sqrt{F(\rho,\sigma)}\le d_{\Tr}(\rho,\sigma)\le\sqrt{1-F(\rho,\sigma)}.
\end{align}
In particular, when both $\rho$ and $\sigma$ are pure, the equality holds in the second inequality.
Both measures satisfy the data processing inequality. That is, for any channel $\map{A}$ taking $\rho$ or $\sigma$ as input, we have
\begin{align}
    F(\rho,\sigma)&\le F(\map{A}(\rho),\map{A}(\rho))\\
    d_{\Tr}(\rho,\sigma)&\ge d_{\Tr}(\map{A}(\rho),\map{A}(\sigma)).
\end{align}
In particular, both measures are unitary-invariant. That is, both equalities hold when $\map{A}$ is a unitary channel.

For two quantum channels $\map{A},\map{B}$ with input space $\spc{H}$ and output space $\spc{H}'$, the diamond distance between $\map{A}$ and $\map{B}$ is defined as:
\begin{align}\label{def:diamonddist}
    d_{\diamond}(\map{A},\map{B}):=\sup_{|\psi\>\in\spc{H}\otimes\spc{H}_A}d_{\Tr}\left(\map{A}\otimes\map{I}_A(\psi),\map{B}\otimes\map{I}_A(\psi)\right),
\end{align}
where $\spc{H}_A\simeq\spc{H}$ is an ancillary space. For two unitary channels $\map{U}$ and $\map{V}$, the diamond norm is well-captured by the operator norm (cf.~\cite[Proposition 1.6]{haah2023query}):
\begin{align}\label{bounddiamond}
    \frac12\min_{\varphi}\|e^{i\varphi}U-V\|_{\infty}\le d_{\diamond}(\map{U},\map{V})\le\min_{\varphi}\|e^{i\varphi}U-V\|_{\infty}
\end{align}
where $\|\cdot\|_{\infty}$ is the operator norm (i.e., the largest singular value).

\subsection{Properties of $N$-copy states.}\label{subsec:ncopy}
In general, we could replace a state $\psi$ with another state $\psi'$ in quantum information processing tasks, if their trace distance is negligible. In the multi-copy regime, we can derive such a similar criterion: 
\begin{lemma}\label{lemma:Ncopyfidelity}
    If $d_{\Tr}(\psi,\psi')\ll 1/\sqrt{N}$ for two pure states $|\psi\>,|\psi'\>$, then the trace distance between $|\psi\>^{\otimes N}$ and $|\psi'\>^{\otimes N}$ vanishes as in the large $N$ limit.
\end{lemma} 
This property can be shown using the relation between the trace distance and the fidelity for pure states (i.e., $F=1-d_{\Tr}^2$) and the multiplicativity of fidelity. 
We will make frequent use of it in later sections. 

When $d_{\Tr}(\psi,\psi')\gg 1/\sqrt{N}$, $d_{\Tr}(\psi^{\otimes N},\psi'^{\otimes N})$ goes to one as $N$ becomes larger. There are multiple ways of seeing this. One is to notice that by tomography using $N$ copies of $|\psi\>$ one can construct a confidence region of radius $O(1/\sqrt{N})$ that contains $|\psi\>$. The confidence region of $|\psi\>$ does not overlap with that of $|\psi'\>$ since the distance between the two states is much larger than the radii. In this way, the two states can be almost perfectly distinguished, and Helstrom's theorem on two-state discrimination \cite{helstrom1969quantum} guarantees that they must have nearly unit trace distance.

Speaking of $N$-copy compression, it is natural to ask whether tomography is enough. That is if one can simply measure $\psi^{\otimes N}$ and store the estimate of $\psi$ in a fully classical memory. 
Since the tomography error scales as $1/\sqrt{N}$, it is not immediate from the previous discussion whether this is going to work.
Somewhat surprisingly, it has been shown that tomography does not work no matter how large the classical memory is:
\begin{prop}[\cite{yang2018titcompression}]\label{prop:noclassicalcompression}
Let $(\map{E}_N,\map{D}_N)$ be any $N$-copy compression protocol for $n$-qubit states. If the protocol uses a fully classical memory, then the compression error will not vanish in the large $N$ limit, no matter how large the memory is.
\end{prop}
Therefore, the task of $N$-copy state compression cannot be trivially reduced to quantum state tomography.
In Appendix \ref{subsec:classicalmemo}, we will show the same result for shallow-circuit states.

\section{Efficient local parameterization of shallow-circuit states}\label{sec:localparameterization}
In Section \ref{sec:compression} we will see that, to achieve faithful compression, we need to pin the shallow-circuit state $|\psi_{\sc}\>$ down to a small local region parameterised by up to $\poly(n)$ real parameters. We name such a task as \emph{local parameterization of quantum states}, which is achieved by measuring copies of $|\psi_{\sc}\>$ and then constructing the local model. However, it is more challenging than the task of quantum state tomography, as one has to not only identify a neighborhood of the state but also parameterise this neighborhood efficiently.

\subsection{Local parameterization of quantum states}
We begin by giving a general definition of the task.
\begin{defi}[$\epsilon$-rotation]
A unitary gate $\map{U}^\epsilon(\cdot)=U^\epsilon(\cdot)(U^\epsilon)^\dag$ is called an $\epsilon$-rotation if $d_{\diamond}(\map{U}^\epsilon,\map{I})\le\epsilon$.
\end{defi}
As the name suggests, an $\epsilon$-rotation can be expressed as a short-time unitary evolution:
\begin{lemma}[Local parameterization of $\epsilon$-rotations]\label{lemma:epsrotation}
An $\epsilon$-rotation $W^\epsilon$ can always be represented as $W^\epsilon=e^{-i2\epsilon H}$ for some Hermitian $H$ with $\|H\|_{\infty}\le 1$.
\end{lemma}
\begin{proof}
    Consider a generic $\epsilon$-rotation $W$. We can always cast it in the form $W=\exp(-i\tilde{\epsilon}H)$ for some Hermitian operator $\|H\|_\infty=1$. By the equivalence of the diamond norm and the unitary operator norm [cf.~Eq.~(\ref{bounddiamond})], we have
\begin{align}
  \epsilon\ge d_{\diamond}(\map{W},\map{I})&\ge\min_{\varphi}\|e^{i\varphi}I-e^{-i\tilde{\epsilon}H}\|_{\infty}=2\sin\left(\frac{\tilde{\epsilon}}{2}\right)\ge \frac{\tilde{\epsilon}}{2}.
\end{align}
\end{proof}
\begin{defi}[Local parameterization of quantum states]
    Given a set $\set{S}=\{|\psi\>\}$ of $n$-qubit quantum states and a fixed quantum state $|\psi_0\>=U_0|0\>^{\otimes n}$ with the guarantee that $d_{\Tr}(\psi,\psi_0)\le\epsilon$ for every $|\psi\>\in\set{S}$, the task of local parameterization of $\set{S}$ is to identify a set $\set{W}_{\rm loc}$ of $O(\epsilon)$-rotations, such that for every $|\psi\>\in\set{S}$ there exists a $W\in\set{W}_{\rm loc}$ with $|\psi\>=U_0W|0\>^{\otimes n}$. Moreover, the local parameterization is \emph{efficient} if $\set{W}_{\rm loc}$ can be parameterized with $\poly(n)$ local parameters, i.e., if there exists an on-to function $\map{P}:\set{E}^{f}\to \set{W}_{\rm loc}$ with $\set{E}\subset\R$ being an $O(\epsilon)-$ length interval and $f\in\poly(n)$.
\end{defi}
At first sight, it is rather straightforward to find the parameterization: One could just apply the inverse of the fixed unitary $U_0$ on every $|\psi\>$. Then the similarity of $U_0^\dag|\psi\>$ and $|0\>^{\otimes n}$ would guarantee the existence of an $O(\epsilon)$-rotation, as shown in the following lemma:
\begin{lemma}\label{lemma:localization}
    If $d_{\Tr}(\map{V}(\phi),\phi)\le\epsilon$ for a pure state $|\phi\>$ and some unitary $\map{V}$, then there exists a $(4\epsilon)$-rotation such that $\map{V}(\phi)=\map{V}^\epsilon(\phi)$.
\end{lemma}
\begin{proof}
    A desired small rotation $V^\epsilon$ can be explicitly constructed. Note that $|\<\phi|V|\phi\>|\ge\sqrt{1-\epsilon^2}$ by Eq.~(\ref{fuchsvandegraaff}). We define $|\phi^\perp\>:=(V|\phi\>-\<\phi|V|\phi\>|\phi\>)/\|V|\phi\>-\<\phi|V|\phi\>|\phi\>\|$. Consider a unitary $V^\epsilon$ that acts as
    \begin{align}
        \left(V^\epsilon\right)_{\phi,\phi^\perp}=\left(\begin{matrix}\<\phi|V|\phi\> & e^{2i\arg(\<\phi|V|\phi\>)}\|V|\phi\>-\<\phi|V|\phi\>|\phi\>\| \\ -\|V|\phi\>-\<\phi|V|\phi\>|\phi\>\| & \<\phi|V|\phi\>\end{matrix}\right)
    \end{align}
    in the two-dimensional subspace spanned by $\{|\phi\>,|\phi^\perp\>\}$ and as $V^\epsilon|\phi'\>=e^{i\arg(\<\phi|V|\phi\>)}|\phi'\>$ for any $|\phi'\>$ in the complementary subspace. It is straightforward to verify that $V^\epsilon|\phi\>=|\phi\>$. $V^\epsilon$ has two eigenvalues $e^{i\arg(\<\phi|V|\phi\>)\pm i\eta}$ with $\eta=\arccos(|\<\phi|V|\phi\>|)$ and all other eigenvalues are $e^{i\arg(\<\phi|V|\phi\>)}$. Therefore, we have
    \begin{align}
        d_{\diamond}(\map{V}^\epsilon,\map{I})&\le 2\min_{\varphi}\|e^{i\varphi}I-V^\epsilon\|_{\infty}\\
        &\le 2\|e^{i\arg(\<\phi|V|\phi\>)}I-V^\epsilon\|_{\infty}\\
        &=2|e^{i\arccos(|\<\phi|V|\phi\>|)}-1|\\
        &\le 4\epsilon.
    \end{align}
    Here $\|\cdot\|_{\infty}$ denotes the operator norm, and the first inequality comes from Eq.~(\ref{bounddiamond}), the equivalence of the diamond norm and the unitary operator norm.
\end{proof}

The above ``vanilla" local parameterization of quantum states appears a simple task. However, making it parameter-efficient is not a trivial task. The local parameterization in Lemma \ref{lemma:localization} consists of generic $n$-qubit unitaries and has exponentially many degrees of freedom.

In practice, it is more favorable to seek an efficient localization. A hope is when the set of states under consideration has certain structure. For instance, when $\set{S}$ consists of states generated by acting a low-dimensional unknown quantum gate upon a fixed state.
If all such states in $\set{S}$ are close to the original fixed state, it is quite natural to speculate that we could probably replace this unknown unitary by an $O(\epsilon)$-rotation to achieve an efficient local parameterization (since the $O(\epsilon)$-rotation is also low-dimensional).
Somewhat surprisingly, this intuition turns out to be wrong:
\begin{remark}
In general, it is not possible to partially locally parameterize a quantum state. That is, there exists a bipartite pure state $|\psi\>_{AB}$ and a unitary $\map{V}_{B}$ acting on $B$ such that:
\begin{enumerate}
    \item $d_{\Tr}((\map{I}_A\otimes\map{V}_{B})(\psi_{AB}),\psi_{AB})\le\epsilon$;
    \item there is no $O(\epsilon)$-rotation $\map{V}_{B}^\epsilon$ such that $(\map{I}_A\otimes\map{V}_{B})(\psi_{AB})=(\map{I}_A\otimes\map{V}_{B}^\epsilon)(\psi_{AB})$.
\end{enumerate} 
\end{remark}
One such example is $|\psi\>_{AB}=\sqrt{1-\epsilon}|0\>_{A}|00\>_{B}+\sqrt{\epsilon}|1\>_A\frac{|10\>_B+|20\>_B}{\sqrt{2}}\in\spc{H}_2\otimes\spc{H}_3^{\otimes 2}$ and a gate $V_{B}$ that acts trivially on the second qutrit if the first qutrit is in $|0\>$ and maps the second qutrit from $|0\>$ to $|1\>$ ($|2\>$) conditioning on the first qutrit being $|1\>$ ($|2\>$). The output state is 
\begin{align*}
    \sqrt{1-\epsilon}|0\>_A|00\>_B+\sqrt{\epsilon}|1\>_A|B_0\>_{B}
\end{align*}
(with $|B_0\>:=\frac{|11\>+|22\>}{\sqrt{2}}$). If a desired $O(\epsilon)$-rotation $V^\epsilon_{B}$ exists, it has to take the initial marginal state on $B$ to the marginal of the output state. That is:
\begin{align*}
    (1-\epsilon)|00\>\<00|+\epsilon\left(\frac{|10\>+|20\>}{\sqrt{2}}\right)\left(\frac{\<10|+\<20|}{\sqrt{2}}\right)\xrightarrow{V^\epsilon}(1-\epsilon)|00\>\<00|+\epsilon|B_0\>\<B_0|.
\end{align*}
Since a unitary cannot change the eigenvalues, we need
\begin{align*}
    V^\epsilon\left(\frac{|1\>+|2\>}{\sqrt{2}}\otimes|0\>\right)=|B_0\>.
\end{align*}
Note that the initial state is a product state while the final state has entanglement bounded away from zero, this transformation cannot be achieved by a unitary gate close to the identity.

The purpose of the above remark is to justify that our goal -- an efficient local parameterization of shallow-circuit states -- is not as simple as it appears. One of our main results is to rigorously prove that such an efficient local parameterization indeed exists for shallow-circuit states.

\begin{theo}[Efficient local parameterization of shallow-circuit states.]\label{theo:localization}
Given the set $\set{S}_{\sc}=\{|\psi_{\sc}\>\}$ of $n$-qubit pure states generated by depth-$d$ brickwork circuits and a fixed shallow-circuit state $|\psi_0\>=U_0|0\>^{\otimes n}$ of the same structure, with the guarantee that $d_{\Tr}(\psi_{\sc},\psi_0)\le\epsilon$ for every $|\psi_{\sc}\>\in\set{S}_{\sc}$, there exists a set $\set{W}_{\rm loc}$ of $O(\epsilon)$-rotations, such that for every $|\psi_{\sc}\>\in\set{S}_\sc$ there exists a $W\in\set{W}_{\rm loc}$ with $|\psi_\sc\>=U_0W|0\>^{\otimes n}$. 

Here $W=\prod_{j=1}^{n_{\rm gate}}W^{28\epsilon}_j$
with $n_{\rm gate}\le n/d+4$, and each $W^{28\epsilon}_j$ is a $(28\epsilon)$-rotation acting on no more than $8d$ qubits. Among $\{W_j^{28\epsilon}\}$, the maximum of the number of gates that act non-trivially on the same qubit is upper bounded by $n_{\rm overlap}=4$. The set $\set{W}_{\rm loc}$ can be parameterized with $O(n\cdot (2^{16d}/d))$ (real) local parameters.
\end{theo} 
Note that the bound on the number of parameters can be made tighter (but the bound might become less succinct).
Details can be found in the proof; see Appendix \ref{subsec:prooflocalpara}. 

We emphasize once again the importance of Theorem \ref{theo:localization}.
Via tomography, one can only fix the $n$-qubit shallow-circuit state, which lives in a Hilbert space of very high dimension, to a small but still high-dimensional region. With Theorem \ref{theo:localization}, however, one can further delegate the uncertainty to small local rotations that live in a much smaller parameter space.
This achieves an exponential reduction of the parametric dimension and serves as the first and a crucial step of the entire compression protocol.

\subsection{Proof of Theorem \ref{theo:localization}}\label{subsec:prooflocalpara}
\noindent{\it Reduction to two layers.}
A shallow-circuit state $|\psi_{\sc}\>$ is generated by an $n$-qubit, depth-$d$ brickwork shallow circuit.
We divide the brickwork circuit into non-intersecting light-cones, resulting in a two-layer circuit (see Figure~\ref{fig:circuit}):
\begin{align}\label{circuit-def}
    U_{\rm bw} = U^{(2)}U^{(1)} = \left(\bigotimes_{i=1}^{n_1}U_i^{(2)}\right)\left(\bigotimes_{j=1}^{n_2}U_j^{(1)}\right).
\end{align}
More explicitly, we first take gates in the light-cone of the first two qubits as the first gate in the second layer, $U_1^{(2)}$. Then starting from the qubit next to the last qubit in this light-cone, we trace $2d$ qubits' light-cone backwards and define this block of gates as the second gate in the second layer, $U_2^{(2)}$. We repeat this procedure until all qubits are exhausted, and the remaining separated blocks are defined as the first-layer gates, as illustrated in Figure \ref{fig:circuit}. It is easy to see that every (block-)gate acts on no more than $2d$ qubits. There are no more than $n_2\le \frac{n}{2d}+2$
gates in the second layer and, since the first layer has no more gates than the second, the total number of gates $n_{\rm gate}$ in the reduced two-layer circuit is upper bounded as
\begin{align}
    n_{\rm gate}\le \frac{n}{d}+4.
\end{align}

\begin{figure}[bt]
    \centering
    \begin{subfigure}[t]{0.3\textwidth}
         \centering
        \includegraphics[width=\textwidth]{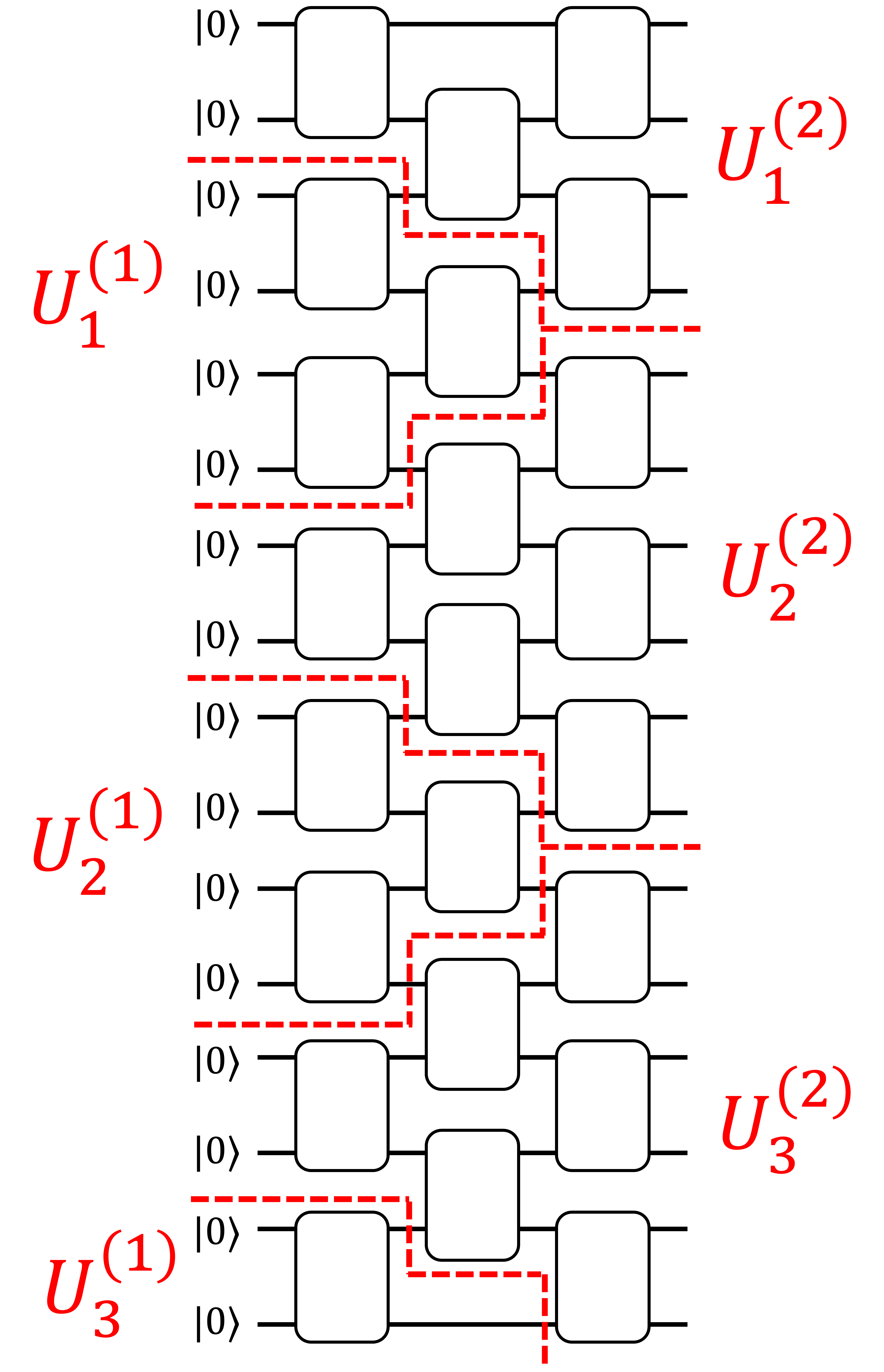}
         \caption{A brickwork shallow-circuit state considered in this work.} 
     \end{subfigure}
     \qquad
     \begin{subfigure}[t]{0.42\textwidth}
         \centering
         \includegraphics[width=\linewidth]{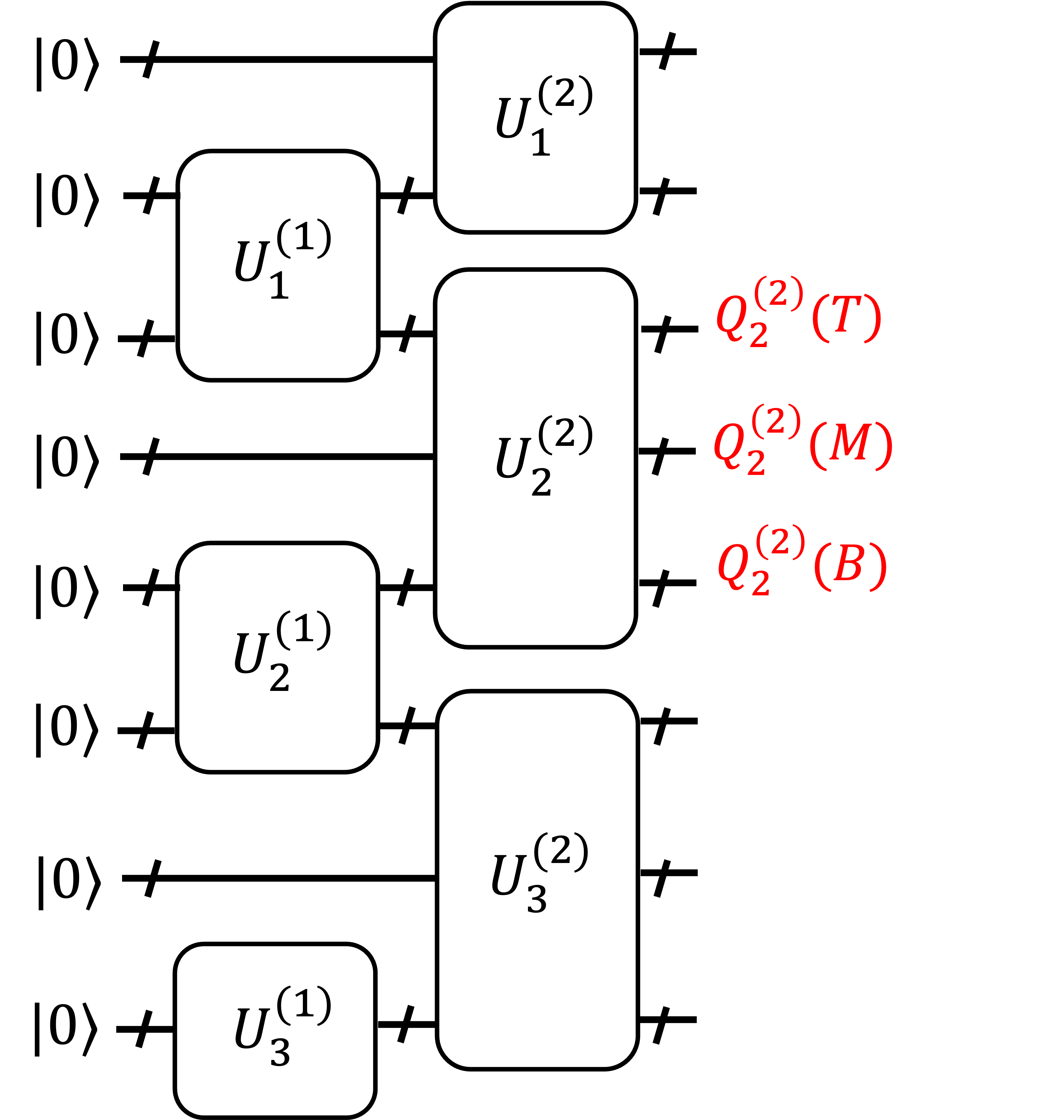}
         \caption{The two-layer reduction of the shallow circuit.}
     \end{subfigure}
    \caption{{\bf Two-layer reduction of a shallow-circuit state.}}
    \label{fig:circuit}
\end{figure}

\medskip

\noindent{\it Grouping qubits.} For convenience of discussion, we group the $n$ qubits into subsets and label them using the following convention. We define a function $Q=Q(U)$ that maps a unitary $U$ to the qubits that it acts non-trivially on. In particular, $Q_j^{(i)}:=Q(U_j^{(i)})$
is the collection of qubits that $U_j^{(i)}$ acts non-trivially on. We further divide $Q_j^{(i)}$ into subsets $Q_j^{(i)}(T)$ (top), $Q_j^{(i)}(M)$ (middle) and $Q_j^{(i)}(B)$ (bottom) as illustrated in Figure~\ref{fig:circuit}. Note that $Q_j^{(i)}(M)=\emptyset$ if $i=1$. For any of these subsets, its cardinality is no more than $d$, and $|Q_j^{(i)}|\le 3d$, i.e., every gate $U^{(i)}_j$ acts non-trivially on no more than $3d$ qubits.
 
By assumption, we are given a shallow-circuit state $|\hat{\psi}_{\sc}\>=\hat{U}_{\rm bw}|0\>^{\otimes n}$ of the same structure. By the same two-layer reduction, we get 
\begin{align}
    \hat{U}_{\rm bw} = \hat{U}^{(2)}\hat{U}^{(1)} = \left(\bigotimes_{i}\hat{U}_i^{(2)}\right)\left(\bigotimes_{j}\hat{U}_j^{(1)}\right).
\end{align}
such that 
\begin{align}
    d_{\Tr}(\psi_{\sc},\hat{\psi}_{\sc})\le\epsilon.
\end{align}
We remark that in practice such a $|\hat{\psi}_{\sc}\>$ can indeed be found via tomography, as discussed later (see Lemma \ref{lemma:tomography}).

\medskip

\noindent{\it Localization of second-layer gates.}
In the following, we shall argue that the desired localization can be achieved via applying the inverse of $\hat{U}_{\rm bw}$ to $|\psi_\sc\>$. This appears rather intuitive, but the local parameterization turns out to be quite non-trivial.

To begin with, we apply the inverse of $\hat{U}^{(2)}$ to $|\psi_\sc\>$. This approximately decouples the $n$-qubit state into the tensor product of $O(n/d)$ states.
The effectiveness is guaranteed by the following lemma:
\begin{lemma}[Approximate decoupling.]\label{lemma:decoupling}
Let $A,A',B,B',C$ be quantum systems. 
Given two bipartite pure states of the form $|\psi\>_{AA'}\otimes|\phi\>_{BB'},|\psi'\>_{AA'}\otimes|\phi'\>_{BB'}$ and a tripartite unitary channel $\map{V}_{ABC}$, if $d_{\Tr}((\map{I}_{A'B'}\otimes\map{V}_{ABC})(\psi_{AA'}\otimes\phi_{BB'}\otimes|0\>\<0|_C),\psi'_{AA'}\otimes\phi'_{BB'}\otimes|0\>\<0|_C)\le\epsilon$, then there exist unitary channels $\map{W}_A,\map{U}_B$ on $A$ and $B$, respectively, such that
\begin{align}
    d_{\Tr}((\map{I}_{A'B'}\otimes\map{V}_{ABC})(\psi_{AA'}\otimes\phi_{BB'}\otimes|0\>\<0|_C),(\map{I}_{A'B'C}\otimes\map{W}_{A}\otimes\map{U}_B)(\psi_{AA'}\otimes\phi_{BB'}\otimes|0\>\<0|_C))\le 3\epsilon.
\end{align}
\end{lemma}
\begin{proof}
    For pure states $\psi,\psi'$, the relation between the fidelity and the trace distance is $F=1-d_{\Tr}^2$ [see Eq.~(\ref{fuchsvandegraaff})], from which we get
    \begin{align}
        F((\map{I}_{A'B'}\otimes\map{V}_{ABC})(\psi_{AA'}\otimes\phi_{BB'}\otimes|0\>\<0|_C),\psi'_{AA'}\otimes\phi'_{BB'}\otimes|0\>\<0|_C)\ge 1-\epsilon^2.
    \end{align}
    Tracing out $A',B,B',C$, we get $F(\rho_A,\rho_A')\ge 1-\epsilon^2$ by data processing, where 
    \begin{align}
        \rho_A&:=\Tr_{A'BB'C}(\map{I}_{A'}\otimes\map{V}_{ABC})(\psi_{AA'}\otimes|0\>\<0|_C\otimes\phi_{BB'})=\Tr_{A'}\psi_{AA'}\\
        \rho_A'&:=\Tr_{A'}\psi_{AA'}'.
    \end{align}
    By Uhlmann's theorem, there exists a unitary channel $\map{W}_A$ such that
    \begin{align}
    F\left(\map{I}_{A'}\otimes\map{W}_A(\psi_{AA'}),\psi'_{AA'}\right)=F(\rho_A,\rho_A')\ge 1-\epsilon^2
    \end{align}
    By Eq.~(\ref{fuchsvandegraaff}), we have
    \begin{align}\label{prelim:temp1}
    d_{\Tr}\left(\map{I}_{A'}\otimes\map{W}_A(\psi_{AA'}),\psi'_{AA'}\right)\le \epsilon.
    \end{align}
    In the same way, we can show that there exists a unitary channel $\map{U}_B$ on $B$ such that
     we have
    \begin{align}\label{prelim:temp2}
    d_{\Tr}\left(\map{I}_{B'}\otimes\map{U}_B(\phi_{BB'}),\phi'_{BB'}\right)\le \epsilon.
    \end{align}
    Combining Eqs.~(\ref{prelim:temp1}) and (\ref{prelim:temp2}), we have 
    \begin{align}
    d_{\Tr}\left((\map{I}_{A'B'C}\otimes\map{W}_A\otimes\map{U}_B)(\psi_{AA'}\otimes|0\>\<0|_C\otimes\phi_{BB'}),\psi_{AA'}'\otimes|0\>\<0|_C\otimes\phi_{BB'}'\right)\le 2\epsilon.
    \end{align}
    Combining the above inequality with the assumption and applying the triangle inequality, we get the desired inequality.
\end{proof}

For any gate $U_j^{(2)}$ in the second layer, applying the inverse of $\hat{U}_j^{(2)}$ we get
$d_{\Tr}\left(\left(\hat{\map{U}}_j^{(2)}\right)^{-1}(\psi_\sc),\left(\hat{\map{U}}_j^{(2)}\right)^{-1}(\hat{\psi}_\sc)\right)\le\epsilon$.
(Note that, to keep the notations succinct, we will abbreviate   $\map{U}\otimes\map{I}$ to $\map{U}$ when there is no risk of confusion.)
Since $\left(\map{U}_j^{(2)}\right)^{-1}(\hat{\psi}_\sc)$ acts trivially on $Q_j^{(2)}(M)$, by Lemma \ref{lemma:decoupling}, there exist two unitaries $W_j^{(T)}$ and $W_j^{(B)}$ acting non-trivially on $Q_j^{(2)}(T)$ and $Q_j^{(2)}(B)$, respectively, such that 
\begin{align}\label{secondlayer:error}
    d_{\Tr}\left(\left(\hat{\map{U}}_j^{(2)}\right)^{-1}(\psi_\sc),\left(\map{W}_j^{(T)}\otimes\map{W}_j^{(B)}\right)\circ\left(\map{U}_j^{(2)}\right)^{-1}(\psi_\sc)\right)\le 3\epsilon.
\end{align}
By the invariance of the trace distance under unitary transformation, we may now remove all other gates and associated qubits in the second layer, which results in the following inequality:
\begin{align}
    d_{\Tr}\left(\left(\hat{\map{U}}_j^{(2)}\right)^{-1}\map{U}_j^{(2)}\circ\left(\map{U}_{j-1}^{(1)}\otimes\map{U}_j^{(1)}\right)(|0\>\<0|_{Q_j}),\left(\map{W}_j^{(T)}\otimes\map{W}_j^{(B)}\right)\circ\left(\map{U}_{j-1}^{(1)}\otimes\map{U}_j^{(1)}\right)(|0\>\<0|_{Q_j})\right)\le 3\epsilon.
\end{align}
for every $j$. Here $Q_j=Q_{j-1}^{(1)}\cup Q_{j}^{(1)}\cup Q_{j}^{(2)}$ and $|0\>\<0|_{Q_j}$ denotes the ground state of these qubits. By Lemma \ref{lemma:localization}, there exists a $(12\epsilon)$-rotation $\map{V}^{12\epsilon}_j$ acting non-trivially on $Q_j$ such that
\begin{align}\label{secondlayeraction}
    \left(\hat{\map{U}}_j^{(2)}\right)^{-1}\map{U}_j^{(2)}\circ\left(\map{U}_{j-1}^{(1)}\otimes\map{U}_j^{(1)}\right)(|0\>\<0|_{Q_j})=\map{V}_j^{12\epsilon}\circ\left(\map{W}_j^{(T)}\otimes\map{W}_j^{(B)}\right)\circ\left(\map{U}_{j-1}^{(1)}\otimes\map{U}_j^{(1)}\right)(|0\>\<0|_{Q_j}).
\end{align}

\begin{figure}[htb]
    \centering
    \includegraphics[width=\linewidth]{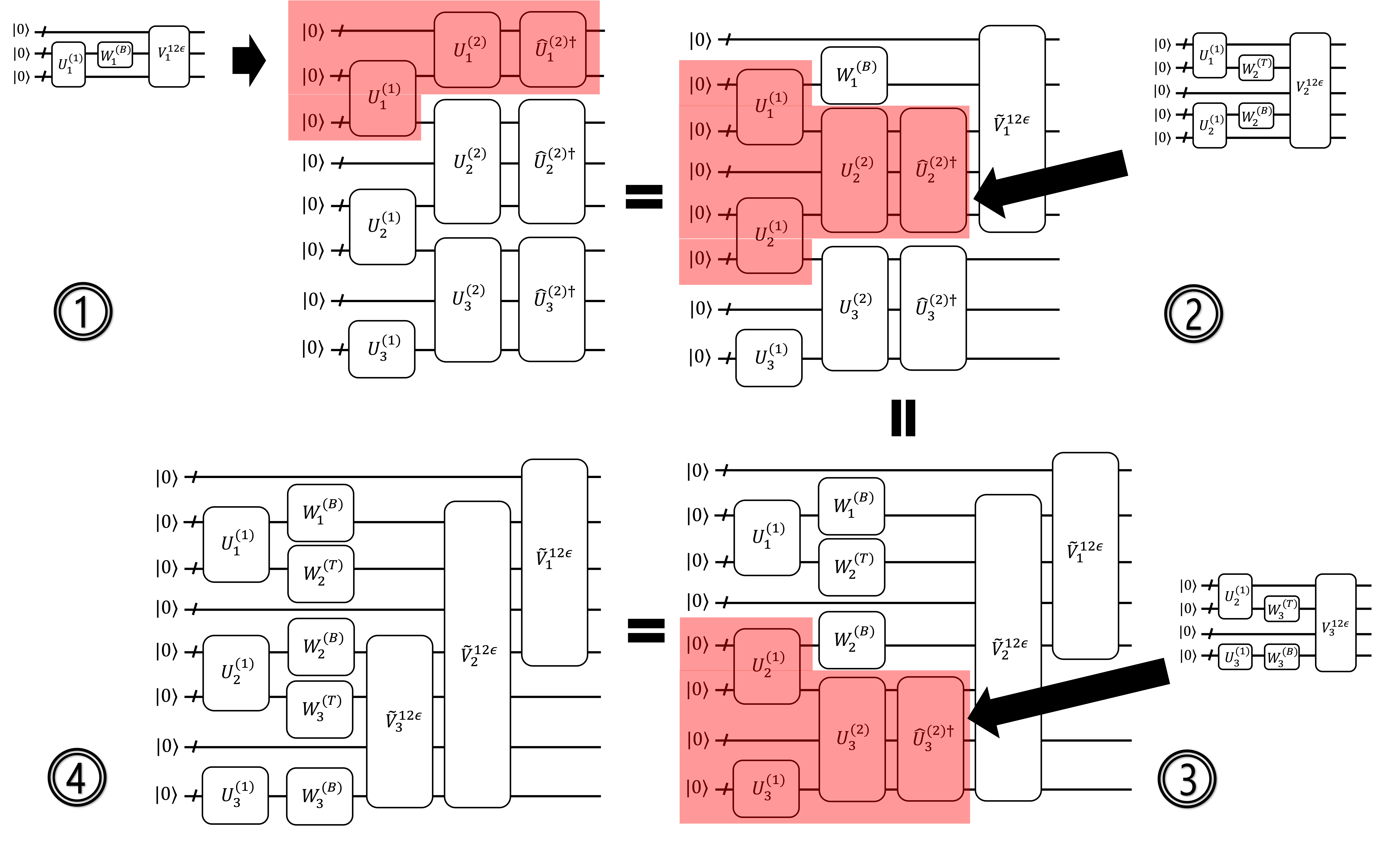}
    \caption{{\bf Gate-wise inversion of the second-layer.} Applying the inverse of $\hat{U}^{(2)}$ to the shallow-circuit state approximately inverts the second-layer gate $U^{(2)}$ up to some unitaries $\{W_j^{(T/B)}\}$ that act only on single registers. The approximation error can be effectively represented by a ladder-shape circuit of $O(\epsilon)$-rotations $\{\tilde{V}_j^{12\epsilon}\}$.}
    \label{fig:secondlayer}
\end{figure}

Now, we see the effect of inverting the second-layer gates. 
As shown in Figure \ref{fig:secondlayer} (step 1), applying the inverse of $\hat{U}^{(2)}$ to $|\psi_{\sc}\>$, we start from analysing the effect of $\hat{U}_1^{(2)}$:
\begin{align}
    &\left(\hat{\map{U}}^{(2)}\right)^{-1}(\psi_\sc)=\left(\bigotimes_{j=2}\left(\hat{\map{U}}^{(2)}_j\right)^{-1}\map{U}^{(2)}_j\right)\circ\left(\left(\hat{\map{U}}^{(2)}_1\right)^{-1}\map{U}^{(2)}_1\circ\map{U}_j^{(1)}(|0\>\<0|_{Q_1})\otimes\left(\bigotimes_{i=2}\map{U}_i^{(1)}\right)(|0\>\<0|_{\bar{Q}_1})\right).
\end{align}
Substituting Eq.~(\ref{secondlayeraction}) into the above equation (Note that $\map{U}_0^{(1)}$ and $\map{W}_1^{(T)}$ are trivial.), we get 
\begin{align}
    &\left(\hat{\map{U}}^{(2)}\right)^{-1}(\psi_\sc)=\left(\bigotimes_{j=2}\left(\hat{\map{U}}^{(2)}_j\right)^{-1}\map{U}^{(2)}_j\right)\circ\left(\map{V}_1^{12\epsilon}\circ\map{W}_1^{(B)}\circ\map{U}_1^{(1)}(|0\>\<0|_{Q_1})\otimes\left(\bigotimes_{i=2}\map{U}_i^{(1)}\right)(|0\>\<0|_{\bar{Q}_1})\right).
\end{align}
As shown in Figure \ref{fig:secondlayer} (step 2), we can move $V_1^{12\epsilon}$ rightwards in the circuit (and leftwards in the equation):
\begin{align}
    &\left(\hat{\map{U}}^{(2)}\right)^{-1}(\psi_\sc)=\tilde{\map{V}}_1^{12\epsilon}\circ\left(\bigotimes_{j=2}\left(\hat{\map{U}}^{(2)}_j\right)^{-1}\map{U}^{(2)}_j\right)\circ\left(\map{W}_1^{(B)}\circ\map{U}_1^{(1)}(|0\>\<0|_{Q_1})\otimes\left(\bigotimes_{i=2}\map{U}_i^{(1)}\right)(|0\>\<0|_{\bar{Q}_1})\right).
\end{align}
Noticing that all remaining gates in the second but $U_2^{(2)}$ commutes with $V_j^{12\epsilon}$ and any $\epsilon$-rotation remains an $\epsilon$-rotation under unitary conjugation, we have $\tilde{V}^{12\epsilon}_1=\left(\map{U}^{(2)}_2\right)^{\dag}\map{U}^{(2)}_2(V^{12\epsilon}_1)$, which is a $(12\epsilon)$-rotation acting on $Q_1^{(2)}\cup Q_1^{(1)}\cup Q_2^{(2)}$.

Repeating the above procedure for every gate in the second layer, we obtain:
\begin{align}\label{secondlayer:localization}
    &\left(\hat{\map{U}}^{(2)}\right)^{-1}(\psi_\sc)=\left(\prod_{j}\tilde{\map{V}}_j^{12\epsilon}\right)\left(\bigotimes_{i}\left(\left(\map{W}^{(B)}_{i}\otimes\map{W}^{(T)}_{i+1}\right)\map{U}^{(1)}_i\left(|0\>\<0|_{Q^{(1)}_i}\right)\otimes|0\>\<0|_{Q^{(2)}_i(M)}\right)\right).
\end{align}
Here each $\left(\map{W}^{(B)}_{i}\otimes\map{W}^{(T)}_{i+1}\right)\map{U}^{(1)}_i$ acts on $Q^{(1)}_i$, and each $\tilde{\map{V}}_j^{12\epsilon}$ is a $(12\epsilon)$-rotation acting on (no more than $7d$) qubits $Q^{(1)}_{j-1}\cup Q^{(2)}_j\cup Q^{(2)}_{j+1}$. See Figure~\ref{fig:secondlayer} (step 4) for an illustration.

\medskip

\noindent{\it Localization of first-layer gates.}
The localization of $U_j^{(1)}$ can be achieved by first (approximately) inverting the two overlapping second-layer gates $U_j^{(2)}$ and $U_{j+1}^{(2)}$ using $\hat{U}_j^{(2)}$ and $\hat{U}_{j+1}^{(2)}$. 
Using Eq.~(\ref{secondlayer:error}) twice and applying the triangle inequality, for every gate of the first layer, we obtain the following:
    \begin{align}
        d_{\Tr}\left(\hat{\map{U}}^{(1)}_j(|0\>\<0|_{Q_j^{(1)}}),\left(\map{W}^{(B)}_j\otimes\map{W}^{(T)}_{j+1}\right)\map{U}^{(1)}_j(|0\>\<0|_{Q_j^{(1)}})\right)\le 7\epsilon.
    \end{align}
    Then, we use the inverse of $\hat{U}_j^{(1)}$.
We can now apply Lemma \ref{lemma:localization} and replace each first layer gate (plus the local degree of freedom) by its estimate and a $(28\epsilon)$-rotation $\map{P}_j^{28\epsilon}$:
\begin{align}\label{firstlayer:localization}
       \left(\map{W}^{(B)}_j\otimes\map{W}^{(T)}_{j+1}\right)\map{U}^{(1)}_j(|0\>\<0|_{Q_j^{(1)}})=\hat{\map{U}}^{(1)}_j\circ\map{P}_j^{28\epsilon}(|0\>\<0|_{Q_j^{(1)}}).
\end{align}
Overall, to localize the whole state, we apply $\hat{\map{U}}^{-1}$ on $\psi_{\sc}$. Combining Eq.~(\ref{secondlayer:localization}) with Eq.~(\ref{firstlayer:localization}), we get
\begin{align}
    \left(\hat{\map{U}}\right)^{-1}(\psi_\sc)=\left(\prod_{j}\tilde{\map{V}'}_j^{12\epsilon}\right)\left(\bigotimes_{i}\map{P}_i^{28\epsilon}\right)(|0\>\<0|^{\otimes n}),
\end{align}
where
\begin{align}
    \tilde{V'}_j^{12\epsilon}=\left(\hat{\map{U}}^{(1)}\right)^{-1}(\tilde{V}_j^{12\epsilon})
\end{align}
is a $(12\epsilon)$-rotation acting on (no more than $8d$) qubits $Q^{(1)}_{j-1}\cup Q^{(2)}_j\cup Q^{(2)}_{j+1}\cup Q^{(1)}_{j+1}$ and  $P_i^{28\epsilon}$ is a $(28\epsilon)$-rotation acting on (no more than $2d$) qubits $Q^{(1)}_i$. 

In summary, there are no more than $n/d+4$ rotations in total, since the rotations are in one-to-one correspondence with the two-layer reduced gates. 
Also notice that $\tilde{V}'^{12\epsilon}_j$ and $\tilde{V}'^{12\epsilon}_k$ overlap only if $|j-k|\le 2$. Taking into account $\otimes_j\map{P}_j^{28\epsilon}$, we conclude that the number of rotations acting non-trivially at a qubit is no more than 4. 
Finally, since a rotation on $8d$ qubits can be characterized by $4^{8d}-1$ parameters and there are $O(n/d)$ such rotations, the total number of real parameters of the local model is $O((n/d)\cdot 4^{8d})$.

\section{Quantum local asymptotic normality (Q-LAN).}\label{sec:QLAN}
Generally speaking, taking $N$ copies of a $D(<\infty)$-dimensional state (which can be mixed) as input, a Q-LAN \cite{guctua2006local,guctua2007local,kahn2009local} is a transformation (independent of the input state) that outputs a state close to a multi-mode Gaussian state. The transformation is approximately invertible. Namely, there exists an inverse Q-LAN transformation recovering the original state up to an error vanishing in $N$. Here locality means that the Q-LAN works only if all states under consideration are in a neighborhood with a vanishing radius.
The Q-LAN can reduce the information processing of finite-dimensional quantum systems (when multiple copies are at hand) to that of Gaussian systems and has been an ingredient to achieve the ultimate precision limit of multiparameter quantum state estimation \cite{gill2013asymptotic,yang2019cmpattaining,girotti2023optimal}.

For our task of compression, we would like to reduce $N$ copies of a shallow-circuit state to a Gaussian state via the Q-LAN.
Problems with the existing Q-LANs are: $1)$ $D=2^n$ is too large; $2)$ $D$ was assumed to be fixed but here $D$ is varying ($n$-dependent); $3)$ the parametrization of the $D$-dimensional state is fixed to be a standard form (see \cite{kahn2009local}); we need to convert our $O(\epsilon)$-rotation parameterization to the standard parametrization.

In the following, we bridge the gap between the existing results and the setting of this work by finding a \emph{new} Q-LAN for pure states generated by $O(\epsilon)$-rotations.
For the generality of our results, let us consider a pure state generated by the concatenation of $G$ gates: $|\psi\>=U_G U_{G-1}\cdots U_1|0\>$,
where each of the gates is an $\epsilon$-rotation that acts non-trivially only on a constant number of qubits. The reduced shallow-circuit state after parameter localization (see Appendix \ref{sec:localparameterization}) stands as a special case. 

Explicitly, our Q-LAN concerns the following family of states:
Let $Q=(q_1,\dots,q_G)$ be a sequence of subsets of $[n]$ with each $|q_j|\le \tilde{d}$. Let
\begin{align}
    \set{S}_{N,\eta}(Q):=\left\{|\psi(\vec{H})\>^{\otimes N}~:~|\psi(\vec{H})\>=|\psi(H_1,\dots,H_G)\>=\prod_{j=1}^{G}\exp\left(-i\frac{\eta}{\sqrt{N}} H_j\right)|0\>^{\otimes n}\right\}
\end{align}
be the $N$-copy family of $n$-qubit pure states generated by $O(\eta/\sqrt{N})$-time unitary evolutions with a fixed circuit template $Q$, where every $H_j$ is a Hermitian operator  acting trivially on qubits that are not in $q_j$ and $\|H_j\|_{\infty}\le 1$. 
Denote by $\set{B}(Q)=\cup_{j=1}^G\left\{\vec{k}\in2^{[n]}~:~k_i=0~\forall i\not\in q_j\right\}$, which keeps track of all $n$-qubit strings that could possibly be ``generated" by acting $\sum_j H_j$ upon $|0\>^{\otimes n}$, and by
\begin{align}\label{def:cohstate}
    |\vec{u}(\vec{H})\>_{\rm coh}=\bigotimes_{\vec{k}\in\set{B}(Q)}|u^{\vec{k}}_1+iu^{\vec{k}}_2\>_{\rm coh}
\end{align}
a coherent state with no more than $G\cdot 2^{\tilde{d}}$ modes, where $|u\>_{\rm coh}:=e^{-|u|^2/}\sum_{m=0}^{\infty}u^m/\sqrt{m!}|m\>$ is a (single-mode) coherent state over an infinite-dimensional Hilbert space with a Fock basis $\{|m\>\}$, $u^{\vec{k}}_1=\Re\left(\eta\<\vec{k}|\sum_{j=1}^G H_j|\vec{0}\>\right)$ (the real part of $\eta\<\vec{k}|\sum_{j=1}^G H_j|\vec{0}\>$), and $u^{\vec{k}}_2=\Im\left(\eta\<\vec{k}|\sum_{j=1}^G H_j|\vec{0}\>\right)$ (the imaginary part of $\eta\<\vec{k}|\sum_{j=1}^G H_j|\vec{0}\>$). The amplitude of each mode of $|\vec{u}(\vec{H})\>_{\rm coh}$ is upper bounded by 
\begin{align}\label{amplitudebound}
    \max_{\vec{k}}|u^{\vec{k}}_1+iu^{\vec{k}}_2|\le \eta n_{\rm overlap}
\end{align}
where $n_{\rm overlap}:=\max_{i\in[n]}\sum_{j:i\in q_j}1$ is the maximum of the number of subsets overlapping at any local site.

\begin{theo}[Q-LAN for shallow-circuit states.]\label{theo:qlan}
When $\eta^2 G^2\ll \sqrt{N}$, there exist quantum channels $\map{V}_N$ (the Q-LAN) and $\map{V}^\ast_N$ (the inverse Q-LAN) such that  any $|\psi(\vec{H})\>\in\set{S}_{N,\eta}(Q)$ can be reversibly converted into a coherent state (\ref{def:cohstate}) with no more than $G\cdot 2^{\tilde{d}}$ modes up to vanishing errors: 
 \begin{align}
        &\lim_{N\to\infty}\sup_{\vec{H}}d_{\Tr}\left(\map{V}_N(\psi(\vec{H}^{\otimes N})),|\vec{u}(\vec{H})\>\<\vec{u}(\vec{H})|_{\rm coh}\right)=0\\
        &\lim_{N\to\infty}\sup_{\vec{H}}d_{\Tr}\left(\psi(\vec{H}^{\otimes N}),\map{V}^\ast_N(|\vec{u}(\vec{H})\>\<\vec{u}(\vec{H})|_{\rm coh})\right)=0.
    \end{align}
\end{theo}

\medskip

\noindent{\it Proof of Theorem \ref{theo:qlan}.}
We first prepare a few preliminary lemmas and then combine them with a varying dimensional Q-LAN to obtain a Q-LAN that converts $n$-qubit shallow-circuit states to $\poly(n)$-mode coherent states. 

\begin{lemma}\label{lemma:linearization}
Let $U=\prod_{j=1}^{G}W(H_j)$ where $W(H_j)=\exp\left(-i\epsilon H_j\right)$. Here each $H_j$ is a Hermitian operator with bounded norm $\|H_j\|_{\infty}\le 1$. Then, we have
\begin{align}
    \left\|U-W\left(\sum_{j=1}^G H_j\right)\right\|_{\infty}=O\left(\epsilon^2G^2\right).
\end{align} 
\end{lemma}
\begin{proof}
    By the triangle inequality and the unitary-invariance of the operator norm, the error can be bounded by a telescoping sum:
    \begin{align}\label{linear-proof-step1}
       \left\|U-W\left(\sum_j H_j\right)\right\|_{\infty}&\le \sum_{i=2}^{G}\left\|W(H_i)W\left(\sum_{j=1}^{i-1}H_j\right)-W\left(\sum_{j=1}^{i}H_j\right)\right\|_\infty.
    \end{align}
    Each error term can be treated with the convergence property of the Zassenhaus formula $e^{A+B}=e^Ae^Be^{-\frac{[A,B]}{2!}}\cdots$ (cf.~Ref.~\cite{suzuki1977convergence}). Noticing that $\|\sum_{j=1}^{i-1} H_j\|_\infty=O(G)$, we have 
    \begin{align}
        \left\|W(H_i)W\left(\sum_{j=1}^{i-1}H_j\right)-W\left(\sum_{j=1}^{i}H_j\right)\right\|_\infty=O(\epsilon^2G).
    \end{align}
    Substituting into Eq.~(\ref{linear-proof-step1}), we get the desired inequality.
\end{proof}
We proceed with a reduction of the unitary's current form into its standard form in Q-LAN:
\begin{lemma}\label{lemma:parameterreduction}
Let $U=\exp(-i\epsilon H)$ be an $n$-qubit unitary. Define an associated unitary by $$\tilde{U}:=\exp\left(-i\epsilon\sum_{\vec{k}} (\Re(h^{\vec{k}})\sigma_x^{\vec{k}}-\Im(h^{\vec{k}})\sigma_y^{\vec{k}})\right),$$ where $h^{\vec{k}}:=\<\vec{k}|H|\vec{0}\>$, $\sigma_x^{\vec{k}}=|\vec{k}\>\<\vec{0}|+|\vec{0}\>\<\vec{k}|$, $\sigma_y^{\vec{k}}=i(|\vec{k}\>\<\vec{0}|-|\vec{0}\>\<\vec{k}|)$ and $\vec{k}$ denotes a binary string of length $n$ with $\vec{0}:=0^n$. $\Re$ and $\Im$ denote the real part and the imaginary part of a complex number, respectively. Defining $|\psi\>:=U|0\>^{\otimes n}$ and $|\tilde{\psi}\>:=\tilde{U}|0\>^{\otimes n}$, we have
    \begin{align}
        d_{\Tr}\left(\psi,\tilde{\psi}\right)=O\left(\epsilon^2\|H\|_{\infty}^2\right).
    \end{align}
\end{lemma}
\begin{proof}
Denote by $\tilde{H}$ the generator of $\tilde{U}$. Since $\tilde{H}\le H$, we have $\|\tilde{H}\|_{\infty}\le\|H\|_{\infty}$. Now by comparing the Talyor series of $|\psi\>$ and $|\tilde{\psi}\>$ in terms of $\epsilon$, we can see that both series agree in the first two terms $|0\>^{\otimes n}$ and $-i\epsilon H|0\>^{\otimes n}$.
By using basic properties of the operator norm, we get
    \begin{align}
        \||\psi\>-|\tilde{\psi}\>\|&=\|U|0\>^{\otimes n}-\tilde{U}|0\>^{\otimes n}\|\\
        &=(\epsilon^2/2)\|(H^2-\tilde{H}^2)|0\>^{\otimes n}\|+O(\epsilon^3 \|H\|_{\infty}^3)\\
        &\le (\epsilon^2/2)\|H^2-\tilde{H}^2\|_\infty+O(\epsilon^3 \|H\|_{\infty}^3)\\
        &\le \epsilon^2\|H\|_{\infty}^2+O(\epsilon^3 \|H\|_{\infty}^3).
    \end{align}
    Combining the above inequality with $\||\psi\>-|\tilde{\psi}\>\|\ge\sqrt{2-2|\<\psi|\tilde{\psi}\>|}$ as well as the relation between the trace distance and the fidelity for pure states (i.e., $F=1-d_{\Tr}^2$) concludes the proof.
\end{proof}

Finally, directly applying Lemmas \ref{lemma:linearization} and \ref{lemma:parameterreduction} yields that any $|\psi(\vec{H})\>$ is $O(\eta^2 G^2/N)$-close (in trace distance) to 
\begin{align}
    |\psi'(\vec{H})\>:=\exp\left(-i\frac{\eta}{\sqrt{N}}\sum_{\vec{k}} (\Re(h^{\vec{k}})\sigma_x^{\vec{k}}-\Im(h^{\vec{k}})\sigma_y^{\vec{k}})\right)|0\>^{\otimes n}
\end{align}
where $h^{\vec{k}}:=\<\vec{k}|\sum_{j=1}^G H_j|\vec{0}\>$.
By Lemma \ref{lemma:Ncopyfidelity}, as long as $\eta^2 G^2\ll\sqrt{N}$, the trace distance between $|\psi(\vec{H})\>^{\otimes N}\in\set{S}_{N,\eta}(Q)$ and $|\psi'(\vec{H})\>^{\otimes N}$ vanishes for large $N$. In this case, we can focus on constructing a Q-LAN for $|\psi'(\vec{H})\>$.

Crucially, we argue that, while there are $2^n$ binary strings $\vec{k}$, at most $\poly(n)$ elements in $\{h^{\vec{k}}\}$ are nonzero. Therefore, the corresponding coherent state has only $\poly(n)$ non-trivial modes.
Indeed, since $H_j$ acts non-trivially on $|q_j|\le \tilde{d}$ qubits, $\<\vec{k}|H_j|\vec{0}\>=0$ unless $\vec{k}\to\vec{0}_{\overline{q_j}}$ when restricting to the complement $\overline{q_j}$ of $q_j$, and there are at most $2^{\tilde{d}}$ such $\vec{k}$. That is, 
\begin{align}\label{Bj}
    \set{B}_j:=\left\{\vec{k}\in2^{[n]}~:~k_i=0\,\forall i\not\in q_j\right\}
\end{align}
and we have $|\set{B}_j|\le 2^{|q_j|}\le 2^{\tilde{d}}$. Denote by $\set{B}(Q)=\cup_{j=1}^G \set{B}_j$.
By the union bound, the number of nonzero $h^{\vec{k}}$, upper bounded by $|\set{B}(Q)|$, cannot exceed $G2^{\tilde{d}}$. Note that the set $\set{B}(Q)$ is fixed for each $\set{S}_{N,\eta}(Q)$, since they depend only on $\{q_j\}$.
With these notations, we can also bound the amplitude of $|h^{\vec{k}}|$ universally as:
\begin{align}\label{boundhveck}
    |h^{\vec{k}}|\le\max_{i\in[n]}\sum_{j:i\in q_j}\|H_j\|_{\infty}\le n_{\rm overlap},
\end{align}
where $n_{\rm overlap}$ denotes the maximal number of overlapping local Hamiltonians at any qubit.

Noticing that $\sigma_i^{\vec{k}}\sigma_j^{\vec{k}'}|0\>^{\otimes n}=0$ unless $\vec{k}=\vec{k}'$,
each locally rotated state can be expanded as 
\begin{align}\label{local_expand_psi}
    |\psi_{\vec{u}/\sqrt{N}}\>&=|0\>^{\otimes n}-i\sum_{\vec{k}\in\set{B}(Q)}\frac{u_1^{\vec{k}}\sigma_y^{\vec{k}}-u_2^{\vec{k}}\sigma_x^{\vec{k}}}{\sqrt{N}}|0\>^{\otimes n}+O\left(|\set{B}(Q)|\left(\frac{\eta}{\sqrt{N}}\right)^2\right)\\
    &=\sqrt{1+\sum_{\vec{k}\in\set{B}(Q)}\frac{|u^{\vec{k}}|^2}{N}}|\tilde{\psi}_{\vec{u}/\sqrt{N}}\>+O\left(\frac{\eta^2|\set{B}(Q)|}{N}\right)\qquad |\tilde{\psi}_{\vec{u}/\sqrt{N}}\>:=\frac{|0\>^{\otimes n}+\sum_{\vec{k}\in\set{B}(Q)}\frac{u^{\vec{k}}}{\sqrt{N}}|\vec{k}\>}{\sqrt{1+\sum_{\vec{k}\in\set{B}(Q)}|u^{\vec{k}}|^2/N}}
\end{align}
where $u^{\vec{k}}=u_1^{\vec{k}}+iu_2^{\vec{k}}$. Since by assumption $\eta^2|\set{B}(Q)|\le 2^{\tilde{d}}G\eta^2\ll\sqrt{N}$, we have $d_{\Tr}(\psi_{\vec{u}/\sqrt{N}},\tilde{\psi}_{\vec{u}/\sqrt{N}})\ll 1/\sqrt{N}$. By Lemma \ref{lemma:Ncopyfidelity}, we can effectively consider the $N$-copy truncated state $|\tilde{\psi}_{\vec{u}/\sqrt{N}}\>^{\otimes N}$.

So far, we've reduced the task of showing the Q-LAN to showing an isometry $V_N$ exists such that $\<\vec{u}|_{\rm coh}V_N|\tilde{\psi}_{\vec{u}/\sqrt{N}}\>^{\otimes N}$ converges to one in the large $N$ limit. 
We consider a general case where $|\tilde{\psi}_{\vec{u}/\sqrt{N}}\>$ is supported on a $(K+1)$-dimensional basis $\{|\phi_m\>\}_{m=0}^K$ (i.e., where $|\psi_{\vec{u}/\sqrt{N}}\>$ has $2K$ generators). We require $\eta K\ll N^\frac{1}{4}$, where $\eta:=\max_k|u^k|$. It is immediate to check that this is indeed satisfied by our setting (since $K\le|\set{B}(Q)|=O(G)$).

We show an effective Q-LAN exists by explicitly constructing the Q-LAN transformation $\map{V}_N$ and the inverse Q-LAN transformation $\map{V}_N^\ast$ and then showing their effectiveness.
Note that $|\tilde{\psi}_{\vec{u}/\sqrt{N}}\>$ lives in the symmetric subspace $\spc{S}_{K+1}^{(N)}\subset\spc{H}_{K+1}^{\otimes N}$. Now, we define $V_N$ whose action on $|\tilde{\psi}_{\vec{u}/\sqrt{N}}\>^{\otimes N}$ is 
\begin{align}
    V_N:\spc{S}_{K+1}^{(N)}&\to\spc{H}\\
    |\vec{m};N\>&\to|\vec{m}\>.\label{VN}
\end{align}
The action of $V_N$ on the complement subspace can be defined arbitrarily.
Here $\{|\vec{m}\>\}$ is the Fock basis for $K$ modes, and the symmetric state
\begin{align}
    |\vec{m};N\>=\sqrt{\frac{m_1!\cdots m_K!}{[N]_{|\vec{m}|}}}\sum_{\pi\in\grp{P}(N)}U_{\pi}|\phi_1\>^{\otimes m_1}|\phi_2\>^{\otimes m_2}\cdots|\phi_K\>^{\otimes m_K}|\phi_0\>^{\otimes (N-|\vec{m}|)},
\end{align}
where $U_\pi$ is the action of an $N$-permutation $\pi$ in the symmetric group $\grp{P}(N)$, $|\vec{m}|$ denotes the sum of all entries of $\vec{m}$ and $[N]_k:=N\cdot(N-1)\cdots(N-k+1)$.
Note that an isometry is reversible only in its own range. We define:
\begin{align}\label{qlanchannel}
    \map{V}_N(\cdot)&:=V_N(\cdot)V_N^\dag\\
    \map{V}^\ast_N(\cdot)&:=V_N^\dag(\cdot)V_N+\Tr\left((I_{\spc{H}}-V_N^\dag V_N)(\cdot)\right)\psi_0\label{invqlanchannel}
\end{align}
where $\psi_0$ is an arbitrary fixed state.

Our goal is to show
\begin{align}
    \<\vec{u}|_{\rm coh}V_N|\tilde{\psi}_{\vec{u}/\sqrt{N}}\>^{\otimes N}&=\sum_{\vec{m}:|\vec{m}|\le N}\<\vec{u}|_{\rm coh}|\vec{m}\>\cdot\<\vec{m};N|\tilde{\psi}_{\vec{u}/\sqrt{N}}\>^{\otimes N},
\end{align}
converges to one for large enough $N$.
Note that
\begin{align}
    \<\vec{m};N|\tilde{\psi}_{\vec{u}/\sqrt{N}}\>^{\otimes N}&=\left(1+\sum_{j=1}^K\frac{|u^{j}|^2}{N}\right)^{-\frac{N}{2}}\sqrt{\frac{[N]_{|\vec{m}|}}{m_1!\cdots m_K!}}\prod_{k=1}^{K}\left(\frac{u^k}{\sqrt{N}}\right)^{m_k}.
\end{align}
 We therefore have
\begin{align}
    \<\vec{u}|_{\rm coh}V_N|\tilde{\psi}_{\vec{u}/\sqrt{N}}\>^{\otimes N}&=\sum_{\vec{m}:|\vec{m}|\le N}\left(\prod_{k=1}^K\left(e^{-\frac{|u^{k}|^2}{2}}\frac{(\overline{u^k})^{ m_{k}}}{\sqrt{m_{k}!}}\right)\right)\left(1+\sum_{j=1}^K\frac{|u^{j}|^2}{N}\right)^{-\frac{N}{2}}\sqrt{\frac{[N]_{|\vec{m}|}}{m_1!\cdots m_K!}}\prod_{k=1}^{K}\left(\frac{u^k}{\sqrt{N}}\right)^{m_k}\\
    &=e^{-\frac{\sum_j |u^{j}|^2}{2}}\left(1+\sum_{j=1}^K\frac{|u^{j}|^2}{N}\right)^{-\frac{N}{2}}\sum_{\vec{m}:|\vec{m}|\le N}\sqrt{\frac{[N]_{|\vec{m}|}}{N^{|\vec{m}|}}}\left(\prod_{k=1}^K \frac{|u^k|^{2m_{k}}}{m_{k}!}\right).
\end{align}
Noticing that every term in the summation is non-negative, we can lower bound the summation by restricting the range of $\vec{m}$ to  $\set{M}_{\rm res}:=\left\{\vec{m}~:~m_k\le N^{\frac14}(\eta/K)\ \forall k\right\}$. This allow us to split the summation and further simplify its expression:
\begin{align}
    \<\vec{u}|_{\rm coh}V_N|\tilde{\psi}_{\vec{u}/\sqrt{N}}\>^{\otimes N}&\ge e^{-\frac{\sum_j |u^{j}|^2}{2}}\left(1+\sum_{j=1}^K\frac{|u^{j}|^2}{N}\right)^{-\frac{N}{2}}\sum_{\vec{m}\in\set{M}_{\rm res}}\sqrt{\frac{[N]_{|\vec{m}|}}{N^{|\vec{m}|}}}\left(\prod_{k=1}^K \frac{|u^k|^{2m_{k}}}{m_{k}!}\right)\\
    &\ge e^{-\frac{\sum_j |u^{j}|^2}{2}}\left(1+\sum_{j=1}^K\frac{|u^{j}|^2}{N}\right)^{-\frac{N}{2}}\sum_{k=1}^K\sum_{m_k\le N^{\frac14}(\eta/K)}\left(1-\frac{|\vec{m}|}{N}\right)^{-\frac{N}{|\vec{m}|}\cdot\left(-\frac{|\vec{m}|^2}{2N}\right)}\frac{|u^k|^{2m_{k}}}{m_{k}!}.
\end{align}
Since $|\vec{m}|\le N^{\frac14}\eta\ll\sqrt{N}$ by assumption, $(1-|\vec{m}|/N)^{-N/|\vec{m}|}$ converges to $e$. In addition, since $|\vec{m}|^2/(2N)\ll 1$, when $N\gg 1$ the right hand side of the last inequality becomes 
\begin{align}
 \underbrace{e^{-\frac{\sum_j |u^{j}|^2}{2}}\left(1+\sum_{j=1}^K\frac{|u^{j}|^2}{N}\right)^{-\frac{N}{\sum_j |u^{j}|^2}\cdot\left(-\frac{\sum_j |u^{j}|^2}{2}\right)}}_{=:S_1}\underbrace{\sum_{k=1}^K\sum_{m_k\le N^{\frac14}(\eta/K)}\frac{|u^k|^{2m_{k}}}{m_{k}!}}_{=:S_2}.
\end{align}
As $\sum_j|u^j|^2\le K\cdot\eta^2\ll N$, the leading order term of $S_1$  is $e^{-\sum_j |u^j|^2}$.
Noticing that $N^{\frac14}(\eta/K)\gg \eta^2\ge |u^k|^2$ by assumption,
when $N$ is large, $\sum_{m_k\le N^{\frac14}(\eta/K)}\frac{|u^k|^{2m_{k}}}{m_{k}!}$ is close to $e^{|u^k|^2}$, and the leading order term of $S_3$ is thus $e^{\sum_k|u^k|^2}$. In conclusion, $S_1\cdot S_2$ converges to one in the large $N$ limit, and so does $\<\vec{u}|_{\rm coh}V_N|\tilde{\psi}_{\vec{u}/\sqrt{N}}\>^{\otimes N}$.

\qed

\section{A compression protocol for shallow-circuit states.}\label{sec:compression}

\subsection{Preliminaries: truncation and amplification of coherent states.}
As shown by the Q-LAN (Appendix \ref{sec:QLAN}), the $N$-copy states we are interested in are asymptotically equivalent to coherent states. Therefore, we introduce some preliminary protocols on processing coherent states, which will be used as building blocks of the compression protocol for shallow-circuit states.

\begin{lemma}[Amplification of coherent states \cite{caves1982quantum}]\label{lemma:amplification}
    There exists an $\alpha$-independent quantum channel that amplifies a coherent state $|\sqrt{1-\epsilon}\alpha\>$ to $|\alpha\>$ (i.e., with intensity gain $\frac{1}{1-\epsilon}$) with (trace distance) error $O(\epsilon)$. 
\end{lemma}
\begin{lemma}[Compression of coherent states]\label{lemma:truncation}
For every $\delta>0$, there exists an $\alpha$-independent compression protocol that compresses any coherent state $|\alpha\>_{\rm coh}$ with $|\alpha|\le \alpha_{0}$ to an $(e^2\alpha_0^{2})$-dimensional quantum memory with recovery error $O(e^{-\alpha_0^2})$.
\end{lemma}

\begin{proof}
    Consider the photon-number truncation channel $\map{T}_{m_0}(\cdot):=P_{m_0}(\cdot)P_{m_0}+\rho_0\Tr(I-P_{m_0})(\cdot)$, where $m_0\in\N^\ast$, $P_{m_0}$ is the projector on $\Span(\{|m\>\}_{m=0}^{m_0-1})$ and $\rho_0$ is a fixed state on $\Span(\{|m\>\}_{m=0}^{m_0-1})$. Applying $\map{T}_{m_0}$ to a coherent state $|\alpha\>_{\rm coh}$ compresses it into a memory of $m_0$ qubits, and the error of compression (with respect to the trivial decoding) is
    \begin{align}
        d_{\Tr}\left(|\alpha\>\<\alpha|_{\rm coh},\map{T}_{m_0}(|\alpha\>\<\alpha|_{\rm coh})\right)&\le \sqrt{1-\<\alpha|_{\rm coh}\map{T}_{m_0}(\alpha_{\rm coh})|\alpha\>_{\rm coh}}\\
        &\le\sqrt{1-(\<\alpha|_{\rm coh}P_{m_0}|\alpha\>_{\rm coh})^2}\\
        &\le\sqrt{2{\rm Poistail}(|\alpha|^2,m_0)},
    \end{align}
    having used Eq.~(\ref{fuchsvandegraaff}), the Fuchs-van de Graaff inequalities. Here ${\rm Poistail}(\lambda,k):=\sum_{m\ge k}e^{-\lambda}\lambda^m/m!$ denotes the tail of a Poisson distribution. By the Chernoff bound, as long as $m_0>|\alpha|^2$, the tail is bounded as
    \begin{align}
        {\rm Poistail}(|\alpha|^2,m_0)\le e^{-|\alpha|^2}\left(\frac{e|\alpha|^2}{m_0}\right)^{m_0}\le\left(\frac{e|\alpha|^2}{m_0}\right)^{m_0}.
    \end{align}
    Taking $m_0 = (e\alpha_0)^2$ and combining the above inequalities conclude the proof.
\end{proof}

\subsection{Efficient tomography of shallow-circuit states.}\label{subsec:tomography}
To prepare for the compression, we introduce some recent results on tomography of shallow-circuit states, which is achieved by constructing a dense-enough covering of shallow-circuit states and then performing a hypothesis testing to find a member of the net close enough to $|\psi_{\sc}\>$.

An $\epsilon$-covering of a set $\set{S}$ of unitary gates is a set $\hat{\set{S}}$ such that for every $U\in\set{S}$ there exists $\hat{U}\in\hat{\set{S}}$ satisfying $d_{\diamond}(U,\hat{U})\le\epsilon$.
The following lemma (see \cite[Theorem 8]{zhao2023learning} and its proof) gives a small-cardinality covering of low complexity unitaries:
\begin{lemma}[Covering complexity-$G$ unitaries]\label{lemma:covering}
Let $\set{U}^G$ be the set of $n$-qubit unitaries that can be implemented by $G$ two-qubit gates. For any $\epsilon\in(0,1]$, there exists an $\epsilon$-covering $\set{U}^G(\epsilon)$ of $\set{U}^G$ whose cardinality is bounded as
\begin{align}
    \log_2\left|\set{U}^G(\epsilon)\right|\le 32G\log_2\left(\frac{12G}{\epsilon}\right)+2G\log_2 n.
\end{align}
Moreover, the covering is tomographic, meaning that each $\hat{U}$ in the covering consists of a sequence of $G$ gates $(\hat{U}_1,\dots,\hat{U}_G)$.
\end{lemma}

Given a covering of unitaries (and, consequently, a covering of the generated states) and copies of an unknown state, a hypothesis testing protocol can be run to find a member of the covering that is close enough to the unknown state, as long as the covering is not too big. We will use the following result \cite{buadescu2021improved,zhao2023learning} on a hypothesis testing which uses classical shadow \cite{huang2020predicting} to improve its performance:
\begin{lemma}[Hypothesis selection by classical shadow]\label{lemma:hypothesis}
Let $0<\epsilon,\delta<1/2$. Given access to $\rho^{\otimes M}$ and classical descriptions of $m$ hypothesis states $\sigma_1,\dots,\sigma_m$, there exists a quantum algorithm that selects $\sigma_k$ such that $d_{\Tr}(\rho,\sigma_k)\le 3\eta+\epsilon$ with probability at least $1-\delta$, where $\eta=\min_id_{\Tr}(\rho,\sigma_i)$ and
\begin{align}
    \epsilon = O\left(\frac{\log_2(m/\delta)}{\sqrt{M}}\right).
\end{align}
\end{lemma}
 
We substitute our setting (i.e., to distinguish $m=|\set{U}^G(\epsilon)|$ hypotheses) into the above results and obtain the efficient tomography needed for our task. 
Since an $n$-qubit, $d$-depth brickwork circuit of two-qubit gates consists of at most $dn/2$ gates, we pick $G=dn/2$ in Lemma \ref{lemma:covering}, combining with Lemma \ref{lemma:hypothesis} yields the following tomography algorithm to be used later:
\begin{lemma}[Efficient tomography of shallow-circuit states]\label{lemma:tomography}
    For an unknown $n$-qubit (brickwork) shallow-circuit state, there exists a quantum algorithm that finds a state $|\hat{\psi}_{\sc}\>$ in an $\epsilon$-covering of shallow-circuit states such that $d_{\Tr}(\psi_{\sc},\hat{\psi}_{\sc})\le 4\epsilon$ with probability $1-\delta$ while requiring 
    \begin{align}\label{tomographyerror}
        N_0=O\left(\frac{\left(nd\log_2(nd/\epsilon)+\log_2(1/\delta)\right)^2}{\epsilon^2}\right)
    \end{align}
    copies of $|\psi_{\sc}\>$.
\end{lemma}

\subsection{A compression protocol for shallow-circuit states.}\label{subsec:protocol} 
      Finally, we are in place to present our compression protocol for shallow-circuit states:
\begin{algorithm}[H]
  \caption{Compression protocol for $N$-copy shallow-circuit states.}
  \label{protocol:compression}
   \begin{algorithmic}[1]
   \Statex {\em Encoder:}
   \Statex {\bf Input:}  
   \Statex $N$ copies of a $n$-qubit state $|\psi_\sc\>$ generated by a brickwork circuit of depth-$d$; 
    \Statex a configuration of parameters $(\epsilon_0, N_0,\alpha_0)$.
   \Statex {\bf Require:} A classical memory of $M_{\rm c}$ bits and a quantum memory of $M_{\rm q}$ qubits.
   \State Construct a covering $\set{U}_{\sc}(\epsilon_0)$ for brickwork circuits of depth-$d$ (cf.~Lemma \ref{lemma:covering}).
   \State ({\em Tomography.}) Run state tomography with $N_0$ copies of $|\psi_\sc\>$ (cf.~Lemma \ref{lemma:tomography}), which outputs an estimate $|\hat{\psi}_{\sc}\>=\hat{U}|0\>^{\otimes n}$ with $\hat{U}\in\set{U}_{\sc}(\epsilon_0)$. Store $\hat{U}$ in the classical memory.
   \item ({\em Q-LAN.}) Apply first the inverse of $\hat{U}^{\otimes (N-N_0)}$ and then the Q-LAN transformation (cf.~Theorem \ref{theo:qlan}) on the remaining $N-N_0$ copies.
   \item ({\em Amplification.}) Amplify each mode of the resultant multi-mode coherent state: $|z\>\to |\sqrt{N/(N-N_0)}z\>$, i.e., with intensity gain $N/(N-N_0)$.
   \item ({\em Truncation.}) Compress the multi-mode coherent state via truncating each mode to less than $\alpha_0$ photons. Store the resultant state in the quantum memory. 
   \end{algorithmic}
   \vspace*{0.5em}
    \begin{algorithmic}[1]
   \Statex {\em Decoder:} 
   \Statex {\bf Require from the encoder:} both the quantum memory (step 5) and the classical memory (step 2).
   \item Apply the inverse Q-LAN transformation on the state of the quantum memory.
   \item Retrieve $\hat{U}$ from the classical memory and apply $\hat{U}^{\otimes N}$. 
   \item {\bf Output} the final state.
   \end{algorithmic}
\end{algorithm} 

Our main result is that, when $\log_2 N/\log_2 n>32/3$, there exists a configuration $(\epsilon_0,N_0,\alpha_0)$ such that Protocol \ref{protocol:compression} is a faithful compression protocol (i.e., one whose error vanishes as $N\to\infty$):
\begin{theo}[Performance of the shallow-circuit state compression]\label{theo:compression}
Let $\set{U}_\sc$ be the collection of all $n$-qubit pure states generated by depth-$d$ brickwork circuits, with $d$ being a fixed parameter. When $N=\Theta(n^{32/3+\gamma})$ for some $\gamma>0$, for any  $\Delta\in\left(\frac{6}{32+3\gamma},\frac{3}{4}-\frac{18}{32+3\gamma}\right)$, there exist a compression protocol that faithfully compresses $N$ copies of any state from $\set{U}_\sc$ into a hybrid memory of $8(1-2\Delta/3)nd\left(\log_2 N+O\left(\log_2 n\right)\right)$ classical bits and $(n/d+4)2^{8d+1}/3\cdot\left(\Delta\log_2 N+O(\log_2 n)\right)$ qubits.
\end{theo}
The proof is provided in Appendix \ref{subsec:proofcompression}.

The total memory cost is in the order of $n\cdot\log_2 N$. The quantum-to-classical cost ratio (in the leading order) is 
\begin{align}
    r_{\rm q-c}=\frac{2^{8d-2}\Delta}{d^2(3-2\Delta)}
\end{align}
When the number of copies gets larger fixing $n$, i.e., when $\gamma\gg 1$, we may choose $\Delta\to 6/(32+3\gamma)$ and the ratio $r_{\rm q-c}=O(1/\gamma)$, which means that
the hybrid memory consists mainly of classical bits. 

 \subsection{Proof of Theorem \ref{theo:compression}.}\label{subsec:proofcompression}
Here we prove the main theorem on the effectiveness of shallow-circuit state compression. We will argue that the following configuration $(\epsilon_0, N_0,\alpha_0)$ of Protocol \ref{protocol:compression} achieves the desired features:
\begin{align}\label{configure}
    \epsilon_0 = n\cdot d\cdot N^{-\frac12(1-\frac23\Delta)}\qquad N_0=N^{1-\frac{1}{2}\Delta}\qquad \alpha_0=896\sqrt{N}\epsilon_0.
\end{align}

 \medskip

 \noindent{\it Faithfulness.} By data processing and the triangle inequality, we can upper-bound the overall error of the compression protocol by the sum of the errors of each step:
 \begin{align}
     \epsilon_N \le \epsilon_{\rm tomo} + \epsilon_{\rm qlan} + \epsilon_{\rm amp} + \epsilon_{\rm trun},
 \end{align}
 where $\epsilon_{\rm tomo}$ is the error of tomography, $\epsilon_{\rm qlan}$ is the error of Q-LAN, $\epsilon_{\rm amp}$ is the error of amplification, and $\epsilon_{\rm trun}$ is the error of coherent state truncation.
 It is enough to show that each error term vanishes for large $N$.
\begin{itemize}
    \item  
By Lemma \ref{lemma:tomography}, to output a $(4\epsilon_0)$-close estimate of $|\psi_\sc\>$ with confidence $1-\delta$, taking $\delta=\epsilon_0$ we need at least
\begin{align}
    N_0'=O\left(\frac{\left(nd\log_2(nd/\epsilon_0)\right)^2}{\epsilon_0^2}\right)
\end{align}
copies of $|\psi_\sc\>$. Since $\epsilon_0=nd/\sqrt{N^{1-2\Delta/3}}$, we have $N'_0 = O\left(N^{1-2\Delta/3}(\log N)^2\right)$ and thus $N_0=N^{1-\Delta/2}\gg N_0'$, i.e., the number of copies allocated to tomography is sufficient to guarantee the desired accuracy.
Notice that, when the estimate falls outside the confidence region, the distance between the estimate and the real state is upper bounded by one.
Then the tomography error, $\epsilon_{\rm tomo}$, can be upper bounded by the probability that the estimate fails outside the confidence region:
\begin{align}
    \epsilon_{\rm tomo}\le\delta=\epsilon_0,
\end{align}
which vanishes for large $N$.
\item We now focus on the case when $|\psi_\sc\>$ is in the confidence region (i.e., when the tomography is good) and apply Theorem \ref{theo:localization}. Since the confidence region has radius $4\epsilon_0$, $\hat{U}^\dag|\psi_\sc\>=\prod_{j\le n/d+4}W_j^{112\epsilon_0}|0\>^{\otimes n}$, where $\{W_j^{112\epsilon_0}\}$ are a collection of $(112\epsilon_0)$-rotations. The total number of rotations is no more than $n/d+4$ with overlapping number $n_{\rm overlap}=4$. By Lemma \ref{lemma:epsrotation} and Eq.~(\ref{configure}), each rotation may be expressed as  
\begin{align}
    W_j^{112\epsilon_0}=\exp\left\{-i\frac{224ndN^{\Delta/3}}{\sqrt{N}} H_j\right\}
\end{align}
for some $(8d)$-partite Hamiltonian $H_j$ with $\|H_j\|_{\infty}\le 1$.
Since $224ndN^{\Delta/3}\cdot(n/d+4)\ll N^{\frac14}$, the error of Q-LAN, $\epsilon_{\rm qlan}$, vanishes as guaranteed 
by Theorem \ref{theo:qlan}.
\item The error of photon-number truncation (i.e., of compressing the multi-mode coherent state), $\epsilon_{\rm trun}$, can be bounded as follows. First, substituting $n_{\rm overlap}=4$ and $\eta=224ndN^{\Delta/3}$ into Eq.~(\ref{amplitudebound}), the maximum amplitude of the multi-mode coherent state can be bounded as:
\begin{align} 
    \max_{\vec{k}}|u^{\vec{k}}_1+iu^{\vec{k}}_2|\le 896ndN^{\Delta/3}=896\sqrt{N}\epsilon_0.
\end{align} 
By Lemma \ref{lemma:truncation}, truncating each mode at photon number equal to $\alpha_0=896\sqrt{N}\epsilon_0$ [cf.~Eq.~(\ref{configure})] guarantees an $O(e^{-\alpha_0^2})$ error at each mode. 
By Theorem \ref{theo:qlan}, the coherent state has no more than $(n/d+4)2^{8d}$ modes. 
Therefore, the total truncation error is 
\begin{align}
    \epsilon_{\rm trun}=O\left(n2^{8d}e^{-n^2d^2N^{2\Delta/3}}\right),
\end{align} which vanishes for large $N$.
\item The amplification error, $\epsilon_{\rm amp}$, can be bounded by directly substituting the intensity gain $N/(N-N_0)$ into Lemma \ref{lemma:amplification}. The error for each single mode is thus $N_0/N=O(N^{-\frac12\Delta})$, and for the entire state we have  \begin{align}
    \epsilon_{\rm amp}=O\left(n2^{8d}N^{-\frac12\Delta}\right),
\end{align} 
which vanishes for large $N$.
\end{itemize}

 \noindent{\it Memory cost.}
The memory $M$ consists of two parts: a classical memory of $M_{\rm c}$ bits and a quantum memory of $M_{\rm q}$ qubits. 

As stated in Protocol \ref{protocol:compression} (step 6), the classical memory is to store $\hat{U}$, an element of the covering $\set{U}_{\sc}(\epsilon_0)$. By Lemma \ref{lemma:covering} and substituting in $\epsilon_0=n\cdot d\cdot N^{-\frac12(1-2\Delta/3)}$ and $G=nd/2$, the size of the classical memory is bounded as
\begin{align}
    M_{\rm c}=\log_2|\set{U}_\sc(\epsilon_0)|\le 8\left(1-\frac{2\Delta}{3}\right)nd\left(\log_2 N+O\left(\log_2 n\right)\right).
\end{align}
The quantum memory is used to store the truncated multi-mode states. Again, since the truncation for each mode costs no more than $\log_2(e\alpha_0)^2$ qubits (with $\alpha_0=896ndN^{\Delta/3}$) by Lemma \ref{lemma:truncation} and there are no more than $(n/d+4)2^{8d}$ modes, the size of the quantum memory is bounded as
\begin{align}
    M_{\rm q}\le \frac{(n/d+4)2^{8d+1}}{3}\left(\Delta\log_2 N+O(\log_2 n)\right).
\end{align}

\subsection{Optimality of the memory scaling.}\label{subsec:optimal}
Here we show that the memory consumption of Protocol \ref{protocol:compression} is the minimum at least in scaling.
\begin{theo}[Lower bound on the memory cost]\label{theo:optimal}
Any faithful $N$-copy compression for $n$-qubit, depth-$d$, brickwork shallow-circuit states requires a memory of size $\Omega(n\cdot\log_2 N)$.
\end{theo}
\begin{proof}
Any compression protocol $(\map{E}_N,\map{D}_N)$ can be used by two parties Alice and Bob for communicating some classical information: the message $X$ is first encoded in a quantum register $A$ in an $N$-copy shallow-circuit state, which is compressed into a memory $M$ via $\map{E}_{N}$; the memory $M$ is transmitted to Bob, who decodes the $N$-copy state via $\map{D}_N$ and further decodes the message by measuring the state (which is not too relevant here). Therefore, the memory $M$, as the bottleneck of the communication protocol, must be large enough to host the information content of the communicated ensemble. In another word, a good lower bound on the memory size can be established by finding an ensemble of $N$-copy shallow-circuit states with large enough information content.

We construct such an ensemble as follows.
Define the $N$-copy ensemble of $n$ uncorrelated qubits $\set{S}_{\rm q}$ as
\begin{align}
    \set{S}_{\rm q}:=\left\{\d\psi_1\cdots\d\psi_n,\,\left(|\psi_1\>\otimes\cdots\otimes|\psi_n\>\right)^{\otimes N}\right\}.
\end{align}
Here each $|\psi_j\>$ is a single-qubit pure state, and $\d\psi_j$ is a measure of single-qubit pure states induced by the Haar measure of $\grp{SU}(2)$. 
Obviously, states in $\set{S}_{\rm q}$ are $N$-copy shallow-circuit states (of depth one), and any faithful compression protocol for shallow-circuit states should work faithfully for them as well.
The Holevo information of an ensemble $\set{S}=\{p_x, \rho_x\}$, which captures the information content of $\set{S}$, is defined as \cite{holevo1973bounds}:
\begin{align}
    \chi(\set{S}_{\rm q}):=I(X:A)_{\overline{\rho}},
\end{align}
where $X$ is the classical register that stores the classical description of the $N$-copy state,  $A$ is the quantum register that stores the associated state, and $\overline{\rho}:=\sum_x p_x|x\>\<x|_{X}\otimes(\rho_x)_A$.
Here $I(X:A):=H(X)+H(A)-H(XA)$ is the mutual information, where $H(\cdot)$ denotes the von Neumann entropy of quantum states.
Observing that the classical-quantum system can be split into $n$ uncorrelated pairs, we can reduce the Holevo information to
\begin{align}
    \chi(\set{S}_{\rm q})=n\chi(\set{S}_{\rm 1-q}),
\end{align}
with $\set{S}_{\rm 1-q}$ being the single-qubit ensemble 
\begin{align}
    \set{S}_{\rm 1-q}:=\left\{\d\psi,\,|\psi\>^{\otimes N}\right\}
\end{align}
where $|\psi\>$ is a single-qubit pure state and $\d\psi$ is the Haar measure induced measure of pure states. The Holevo information of the single-qubit ensemble can be explicitly evaluated as:
\begin{align}
    \chi(\set{S}_{\rm 1-q})&=H\left(\int(\d\psi)\,\psi^{\otimes N}\right)-\int(\d\psi)\,H\left(\psi^{\otimes N}\right)\\
    &=H\left(\int\d\psi\,\psi^{\otimes N}\right)\\
    &=H\left(\frac{P_{\rm sym}}{N+1}\right)\\
    &=\log_2(N+1),
\end{align}
where $P_{\rm sym}$ is the projector onto the symmetric subspace of $N$ qubits.
The second equality holds since the entropy of a pure state is zero, the third equality holds thanks to Schur's lemma and that the $N$-copy pure states live in the symmetric subspace, and the last equality holds by definition of the von Neumann entropy. Therefore, the Holevo information of interest is
\begin{align}
    \chi(\set{S}_{\rm q})=n\log_2(N+1).
\end{align}

It is immediate that the above Holevo information $\chi(\set{S}_{\rm q})$ has the desired scaling for both $n$ and $N$. What remains is to show that the required memory size is lower bounded by a quantity close to $\chi(\set{S}_{\rm q})$.
First, noticing that in the aforementioned communication protocol
\begin{align}
    X\xrightarrow{}A\xrightarrow{\map{E}_N}M\xrightarrow{\map{D}_N}B
\end{align}
forms a Markov chain, by data processing, we have $I(X:B)\le I(X:M)$.
Since $X$ is classical, we have $I(X:M)=H(M)-H(M|X)\le H(M)\le \log_2 D_M$ with $D_M$ being the dimension of the memory. Combining the two inequalities yields
\begin{align}
    \log_2 D_M\ge I(X:B).
\end{align}
Second, since the compression is faithful, $I(X:B)$ is almost the same as $I(X:A)=\chi(\set{S}_{\rm q})=n\log_2(N+1)$. Leveraging the continuity of the mutual information \cite{alicki2004continuity,wilde2013quantum} (see also \cite[Supplementary Information]{yang2016prlefficient}), we have
\begin{align}
\log_2 D_M\ge (1-2\epsilon_N)n\log_2(N+1)-2h_2(\epsilon_N),
\end{align}
where $\epsilon_N$ is the compression error and $h_2(x):=-x\log_2 x$. Since $\epsilon_N$ vanishes in the large $N$ limit, we conclude that the  cost of \emph{any faithful compression protocol} is $\Omega(n\cdot\log_2 N)$.
\end{proof}

\subsection{Necessity of quantum memory.}\label{subsec:classicalmemo}
In Appendix \ref{subsec:protocol}, we have seen that the compression does not require a fully quantum memory. Even more, under certain conditions, the hybrid memory can be made almost classical. It is thus tempting to ask whether quantum memory is necessary at all, i.e., whether a fully classical memory can fulfill the goal of faithfully compression. Here we give a negative answer by showing that any protocol using only classical memory cannot achieve faithfulness. Note that this fact was first shown in Ref.~\cite{yang2018titcompression} for qudits; see Proposition \ref{prop:noclassicalcompression}. Here we show the same fact for shallow-circuit states by making some modifications to the original proof.
\begin{theo}[Necessity of quantum memory]\label{theo:classicalmemo}
Consider any $N$-copy compression protocol $(\map{E}_N,\map{D}_N)$ that faithfully compresses $n$-qubit shallow-circuit states of depth-$d$. If $(\map{E}_N,\map{D}_N)$ uses a fully classical memory, it cannot have an error that vanishes in $N$.
\end{theo}
It is noteworthy that the result does not put any restriction on the size of the memory, i.e., the compression cannot be faithful no matter how large the classical memory is.

    To show Theorem \ref{theo:classicalmemo}, we define two distance measures of quantum states (cf.~Ref.~\cite{shunlong2004informational}): the quantum Hellinger distance
    \begin{align}
        d_{\rm H}(\rho,\sigma):=\sqrt{2-2\Tr\left(\rho^{\frac12}\sigma^{\frac12}\right)}
    \end{align}
    and the quantum Bures distance
    \begin{align}
        d_{\rm B}(\rho,\sigma):=\sqrt{2-2\Tr\left|\rho^{\frac12}\sigma^{\frac12}\right|}.
    \end{align}
    The quantum Hellinger distance is reated to the trace distance via:
    \begin{align}\label{distHTr}
    d_{\rm H}(\rho,\sigma)\le\sqrt{2d_{\Tr}(\rho,\sigma)}.
    \end{align}
    Besides the general properties of distance measures (e.g., the triangle inequality and data processing),
    we need two additional properties of the two measures in the proof. The first is:
    \begin{lemma}\label{lemma:propBH1}
        For every $\rho$ and $\sigma$, $d_{\rm H}(\rho,\sigma)\ge d_{\rm B}(\rho,\sigma)$. The equality holds if and only if $[\rho,\sigma]=0$.
    \end{lemma}
    The inequality holds as $\Tr|A|\ge\Tr A$ for any square matrix $A$. To see the condition for equality, notice that the two measures coincide if and only if $A=A^\dag$ for $A:=\rho^{\frac12}\sigma^{\frac12}$, i.e., 
    \begin{align}
        \rho^{\frac12}\sigma^{\frac12}=\left(\rho^{\frac12}\sigma^{\frac12}\right)^\dag=\sigma^{\frac12}\rho^{\frac12},
    \end{align}
    which is equivalent to $[\rho,\sigma]=0$. The second property is an approximate continuity as follows:
    \begin{lemma}\label{lemma:propBH2}
        Let $\map{E}$ be a channel sending states on $\map{H}$ to states on $\spc{H}'$, and $\map{D}$ be a channel sending states on $\spc{H}'$ to states on $\spc{H}$. Let $\rho_1$ and $\rho_2$ be states on $\spc{H}$ such that 
        \begin{align}
            [\map{E}(\rho_1),\map{E}(\rho_2)]=0
        \end{align}
        and $d_{\Tr}(\map{D}\circ\map{E}(\rho_i),\rho_i)\le\epsilon$ for $i=1,2$. Then, we have $d_{\rm H}(\rho_1,\rho_2)-d_{\rm B}(\rho_1,\rho_2)\le 2\sqrt{2}\epsilon$.
    \end{lemma}
    \begin{proof}
By the triangle inequality for $d_{\rm H}$, we have
\begin{align}
    d_{\rm H}(\rho_1,\rho_2)&\le d_{\rm H}(\rho_1,\map{D}\circ\map{E}(\rho_1))+d_{\rm H}(\map{D}\circ\map{E}(\rho_1),\map{D}\circ\map{E}(\rho_2))+d_{\rm H}(\map{D}\circ\map{E}(\rho_2),\rho_2)\\
    &\le d_{\rm H}(\map{D}\circ\map{E}(\rho_1),\map{D}\circ\map{E}(\rho_2))+2\sqrt{2\epsilon}.
\end{align}
The last step comes from the combination of the assumption and Eq.~(\ref{distHTr}). Next, by data processing and Lemma \ref{lemma:propBH1}, we get
\begin{align}
    d_{\rm H}(\map{D}\circ\map{E}(\rho_1),\map{D}\circ\map{E}(\rho_2))&\le d_{\rm H}(\map{E}(\rho_1),\map{E}(\rho_2))\\
    &=d_{\rm B}(\map{E}(\rho_1),\map{E}(\rho_2))\\
    &\le d_{\rm B}(\rho_1,\rho_2).
\end{align}
Combining the above inequalities yields the desired statement.
\end{proof}

We are now in place to prove Theorem \ref{theo:classicalmemo}.

\medskip

\noindent{\em Proof of Theorem \ref{theo:classicalmemo}.}
Consider  two $n$-qubit states $|\psi_1\>:=|0\>^{\otimes n}$ and 
\begin{align}
    |\psi_2\>:=e^{-i\frac{\sigma_X}{\sqrt{N}}}|0\>\otimes|0\>^{\otimes(n-1)}.
\end{align}
where $\sigma_X:=|0\>\<1|+|1\>\<0|$ is the Pauli-$X$ operator.
Obviously, both states are in the set of shallow-circuit states $\set{S}_{\sc}$.
Moreover, both are in a subset $\set{S}'_{\sc}$ of $\set{S}_{\sc}$ where the states are within an $O(1/\sqrt{N})$-radius neighborhood around $|0\>^{\otimes n}$. The Q-LAN has vanishing error for states in $\set{S}'_\sc$. In particular, by Theorem \ref{theo:qlan}, we have
\begin{align}\label{qlan:G1G2}
    \max_{j=1,2}\left\{d_{\Tr}\left(\map{V}_N(\psi_j^{\otimes N}),G_j\right),d_{\Tr}\left(\psi_j^{\otimes N},\map{V}_N^\ast(G_j)\right)\right\}=:\epsilon_{{\rm Q-LAN},N}\qquad\lim_{N\to\infty}\epsilon_{{\rm Q-LAN},N}=0.
\end{align}
Here $|\psi_2\>$ is mapped by the Q-LAN to
\begin{align}
    |G_2\>:=|1\>_{\rm coh}\otimes |{\rm vac}\>
\end{align}
where $|1\>_{\rm coh}$ is a single-mode coherent state with unit amplitude and $|{\rm vac}\>$ denotes the vacuum state of the remaining modes. $|\psi_1\>$ is mapped to  
\begin{align}
    |G_1\>:=|0\>_{\rm coh}\otimes |{\rm vac}\>
\end{align}
where $|0\>_{\rm coh}$ is the single-mode vacuum state. 
Using their definitions, the quantum Hellinger distance and the quantum Bures distance between $G_1$ and $G_2$ can be explicitly evaluated as
\begin{align}
    d_{\rm H}(G_1,G_2)=\sqrt{2(1-e^{-1/2})}\qquad d_{\rm B}(G_1,G_2)=\sqrt{2(1-1/e)},
\end{align}
and thus
\begin{align}\label{distdiff}
    d_{\rm H}(G_1,G_2)-d_{\rm B}(G_1,G_2)>0.167.
\end{align}

For any compression protocol $(\map{E}_N,\map{D}_N)$ for shallow-circuit states, we consider its error on $\psi_1$ and $\psi_2$:
\begin{align}
    \epsilon_N:=\max_{j=1,2}d_{\Tr}({\map{D}}_N\circ{\map{E}}_N(\psi_j),\psi_j).
\end{align}
First, via the Q-LAN, we can define a compression protocol $(\tilde{\map{E}}_N,\tilde{\map{D}}_N)$ for multi-mode coherent states as
\begin{align}
    \tilde{\map{E}}_N:=\map{E}_N\circ\map{V}^\ast_N\qquad\tilde{\map{D}}_N:=\map{V}_N\circ\map{D}_N,
\end{align}
where $\map{V}_N$ and $\map{V}_N^\ast$ are the Q-LAN and the inverse Q-LAN, respectively (cf.~Appendix \ref{sec:QLAN}). 

If the original protocol $(\map{E}_N,\map{D}_N)$ uses a fully classical memory, by definition, so does the new protocol $(\tilde{\map{E}}_N,\tilde{\map{D}}_N)$. We thus have
\begin{align}
    \left[\tilde{\map{E}}_N(G_1),\tilde{\map{E}}_N(G_2)\right]=0.
\end{align}
Applying Lemma \ref{lemma:propBH2}, we get
\begin{align}\label{contra1}
d_{\rm H}(G_1,G_2)-d_{\rm B}(G_1,G_2)\le 2\sqrt{2\tilde{\epsilon}_N},
\end{align}
where $\tilde{\epsilon}_N$ denotes the error of $(\tilde{\map{E}}_N,\tilde{\map{D}}_N)$.
Next, we bound $\tilde{\epsilon}_N$ by using the error of the original protocol.
By data processing and the triangle inequality, we get:
\begin{align}
    \tilde{\epsilon}_N:=\max_{j=1,2}d_{\Tr}(\tilde{\map{D}}_N\circ\tilde{\map{E}}_N(G_j),G_j)&=\max_{j=1,2}d_{\Tr}(\map{V}_N\circ\map{D}_N\circ\map{E}_N\circ\map{V}_N^\ast(G_j),G_j)\\
    &\le \max_{j=1,2}\left(d_{\Tr}(\map{V}^\ast_N(G_j),\psi_j)+d_{\Tr}(\map{D}_N\circ\map{E}_N(\psi_j),\psi_j)+d_{\Tr}(\map{V}_N(\psi_j),\psi_j\right)\\
    &\le 2\epsilon_{{\rm Q-LAN},N}+\epsilon_N.\label{contra2}
\end{align}
Combining Eqs. (\ref{qlan:G1G2}), (\ref{contra1}) and (\ref{contra2}) and taking the large $N$ limit, we get
\begin{align}
    d_{\rm H}(G_1,G_2)-d_{\rm B}(G_1,G_2)\le 2\sqrt{2\lim_{N\to\infty}\epsilon_N}.
\end{align}
Finally, combining the above inequality with Eq.~(\ref{distdiff}), we get
\begin{align}
    \lim_{N\to\infty}\epsilon_N>0.003.
\end{align}
Therefore, the original compression protocol $(\map{E}_N,\map{D}_N)$ is not faithful.

\qed

\subsection{Time complexity and gate complexity of the compression.}\label{subsec:computationcomplexity}

As shown in the description of Protocol \ref{protocol:compression}, the compression consists of several steps. Their execution times are analyzed in the following:
\begin{itemize}
    \item {\it Tomography.} The tomography subroutine \cite{zhao2023learning} we plan to use has a gate complexity that is at least exponential in $G$, the number of gates in the generating circuit of the state. Instead of directly executing the tomography subroutine \cite{zhao2023learning}, we make it more computationally efficient by using the structure of shallow circuits. Since the brickwork shallow circuit can be reduced into two layers (cf.~Section \ref{sec:localparameterization}), the tomography can be done in just two consecutive steps. First, gates in the first layer are determined (up to local unitaries) with a small error. Since all gates in the same layer commute with each other, they can be determined in parallel, and thus the time complexity is equal to the complexity for a single gate (with an error that vanishes at most as $\poly(n)$). Next, the same procedure is applied to gates in the second layer. Since each gate is of a constant size and there are $O(n)$ gates in each layer, the whole tomography can be done efficiently.  

    \item {\it The Q-LAN and the inverse Q-LAN.} Essentially, one needs to implement any transformation such that Eq.~(\ref{VN}) (and its pseudo-inverse) holds, which is the transformation from a certain subset of computational states to the corresponding symmetric states. 
 
    Here we present a concrete implementation of this step consisting of two subroutines. The first is to apply the same unitary transformation $F$ on each $n$-qubit system that reorders the computational bases by the number of ones they contain: $F$ maps a computational basis state $|x\>$ into $|f_x\>$, where $f_x$ is a binary string such that $f_x\ge f_y$ if there are more ones in $x$ than in $y$. This $F$ is a crucial step of the Schumacher compression and can be achieved by the construction in Ref.~\cite{cleve1996schumacher}, which has time complexity $O(n^3)$ with $O(\sqrt{n})$ ancillary qubits. In total, $N$ uses of $F$ in parallel, and thus this subroutine is both time-efficient and computationally efficient. For the inverse Q-LAN in decoding, $F^\dag$ can be implemented (efficiently) with the same complexity. 
    The purpose of doing so is to inversely order the basis states \emph{by the number of ones} for the convenience of the truncation (see the next step). 
    
    The second subroutine is a basis transformation into the symmetric basis. It is sufficient (but not necessary) to implement the Schur transform (and its inverse) on the $N$ systems. According to Refs.~\cite{harrow2005applications,bacon2006efficient,krovi2019efficient}, the Schur transform can be done with polynomial gate complexity in both $N$ and $\log_2 D$ (with $D:=2^n$ being the single-copy dimension), and thus our Q-LAN subroutines are computationally efficient and, as an immediate consequence, time efficient. The algorithm in Ref.~\cite{krovi2019efficient} outputs the representation register in the Gelfand-Tsetlin basis. More explicitly, it takes states in the computational basis as inputs and outputs states in the Schur basis, which has three registers: the index register (labeled by Young diagrams), the multiplicity register, and the representation register. Since only symmetric states are concerned here, the first two registers are trivial, and the third is spanned by the Gelfand-Tsetlin basis. For symmetric states of $\spc{H}_{D=2^n}^{\otimes N}$ with type $(N_1,\dots,N_{2^n})$, the nonzero subregister of the Gelfand-Tsetlin state reads $|N,N-N_1,N-N_1-N_2,\dots,N_{2^n}\>$. Here $N_j$ corresponds to the number of $|j\>$.
    
    \item {\it Photon-number truncation.} 
    We start from the Gelfand-Tsetlin state which reads $|N,N-N_1,N-N_1-N_2,\dots,N_{2^n}\>$ with $2^n$ sub-registers.
    Crucially, note that now the computational basis has been reordered by the unitary $F$ at the previous step (cf.~the first subroutine), and now $N_j$ corresponds to the original binary string with the $j$-th lowest number of ones.
    By Eq.~(\ref{local_expand_psi}), we can focus only on the first $O(2^d n)$ sub-registers with at most $O(d)$ ones adjacent to each other, i.e., those corresponding to binary strings in $\set{B}_j$ (\ref{Bj}).
    One can simply discard the other $O(2^n)$ sub-registers. The remaining registers (consisting of $O(2^n n)$ sub-registers, each of dimension $N$) $|N,N-N_1,N-N_1-N_2,\dots\>$ can be transformed into a type register $|N_1,N_2,\dots\>$ by consecutive subtraction, which can be done efficiently. 

    Next, we need to do a number truncation on each sub-register ($N$-dimensional), as required in Lemma \ref{lemma:truncation}.
    Since the truncation operations on different modes are independent and can be implemented in parallel. It is sufficient to analyze the time it takes to implement one truncation. 
    

    Suppose one is given a register (consisting of qubits) in a pure state. In our setup, to realize a truncation above (or below, in a similar way) a threshold value $N_0$, it is enough to simply discard all but the last $\lceil \log_2(N_0-1)\rceil$ qubits.
    Indeed, as the state can be represented as 
    \begin{align}
        |\psi\>=|0\cdots0\>\otimes\sum_{b\le N_0}\psi_b|b\>+\sqrt{\epsilon_0}|\psi'\>\in\spc{H}_{A}\otimes\spc{H}_{B},
    \end{align}
    for some normalized vector $|\psi'\>$ orthogonal to $\{|0\cdots0\>\otimes|b\>\}$, where $\spc{H}_B$ corresponds to the last $\lceil\log_2(N_0-1)\rceil$ qubits and $\spc{H}_A$ corresponds to the other qubits, $b$ is the basis index in the binary representation and $\psi_b$ denotes the corresponding amplitude with $\sum_{b\le N_0}|\psi_b|^2=1-\epsilon_0$. Discarding the first sub-register yields $\rho=\sum_{b,b'\le N_0}\psi_b\psi^\ast_{b'}|b\>\<b'|+\Tr_{A}\psi'$, which is a state on $\lceil\log_2(N_0-1)\rceil$ qubits, and appending zero qubits to it yields a state $|0\cdots0\>\<0\cdots0|\otimes\rho$ with at least $(1-\epsilon_0)^2$ fidelity with $\psi$. Therefore, the truncation would have a negligible error if $\epsilon_0$ is small.

    Overall, the truncation can be done time efficiently. Whether it is also gate-efficient depends on the complexity of discarding/resetting a sub-register. As one has to discard $O(2^n)$ sub-registers (in parallel, so the time complexity is low), the truncation step is gate-efficient if the complexity cost of discarding these registers is zero; otherwise, this step is not gate-efficient.

    \item {\it Amplification.}
    For amplification, we need to apply the same single-mode amplifier to each of the $O(n)$ modes. Again, we may simply analyze the efficiency of a single-mode amplifier. The amplification can be realized by the original proposal in Ref.~\cite{caves1982quantum}, where the mode is jointly squeezed together with a vacuum mode by a two-mode squeezing operation. The degree of squeezing is related to the degree of amplification required and is close to one in our case (i.e., the two-mode squeezing operation is close to identity). The two-mode squeezing is a standard technique in quantum optics and can be implemented under nowadays lab conditions. On the other hand, the cost of such an operation on infinite dimensional systems is not directly linked with the notion of gate complexity. 

    Instead, we show that even skipping the amplification step (i.e., applying the identity gate to each mode), the whole protocol may still be faithful since the intensity gain is close to one. 
    To see this, notice that for $g\ge 1$, $_{\rm coh}\<g\cdot z|z\>_{\rm coh}=e^{-\frac{(g-1)^2}{2}|z|^2}$ for two coherent states. Therefore, the amplification error per mode scales as $O\left((g-1)^2\alpha_0^2\right)$ as long as $(g-1)^2\alpha_0^2\ll 1$, where $(\alpha_0)^2$ (as in Protocol \ref{protocol:compression}) upper bounds the intensity per mode. Substituting in the explicit configuration (\ref{configure}) (here $g^2=N/(N-N_0)$), we get that the single-mode error scales as $O(n^2\cdot N^{-\Delta/3})$. Since there are $O(n)$ modes in total, the total error is $O(n^3\cdot N^{-\Delta/3})$. Since $N=\Theta(n^{32/3+\gamma})$, the error is vanishing whenever $\Delta>\frac{27}{32+3\gamma}$. This is achievable when $\gamma$ is large enough (i.e., when $N$ is large enough compared to $n$). We remark that there are also other promising and interesting approaches, e.g., finite-dimensional approximation via spin-coherent states, that are to be tested in future works.
    Overall, the amplification step can also be realized both time-efficiently and gate-efficiently.
\end{itemize}
In conclusion, we have shown that all subroutines of the compression protocol (Protocol \ref{protocol:compression}) can be realized time-efficiently in principle. It can also be realized gate-efficiently, if the gate complexity of discarding qubits can be neglected (see the discussion on the truncation subroutine).
We note that there is definitely room for further improving the algorithm's efficiency, and the minimum complexity of shallow-circuit compression remains an interesting and important question to be explored in future works.

\end{widetext}
 
\end{document}